\documentclass[12pt]{article}
\newcommand{\blind}{0}

\usepackage{graphicx,epsf,subfigure,pstricks,psfrag,pst-node,amsthm,amssymb,amsmath,natbib}
\usepackage{multirow,natbib}
\usepackage{tabularx}
\usepackage{enumerate}
\usepackage{hyperref}
\usepackage{bm}
\usepackage{extarrows, framed}
\usepackage{accents}
\usepackage{enumitem}
\usepackage{pbox}

\newcommand{\argmin}{\arg\!\min}
\newtheorem{theorem}{Theorem}
\newtheorem{lemma}{Lemma}
\newtheorem{corollary}{Corollary}

\newtheorem{assump}{Assumption}

\date{\vspace{-5ex}}
\begin{document}

\def\spacingset#1{\renewcommand{\baselinestretch}%
{#1}\small\normalsize} \spacingset{1}


\if0\blind
{
  \title{\bf Variable selection   via  penalized credible regions  with Dirichlet-Laplace  global-local shrinkage priors}
  \author{Yan Zhang 
\hspace{.2cm}\\
    Department of Statistics, North Carolina State University\\
    and \\
    Howard D. Bondell  \\
    Department of Statistics, North Carolina State University}
  \maketitle
} \fi

\if1\blind
{
  \bigskip
  \bigskip
  \bigskip
  \begin{center}
    {\LARGE\bf Variable selection   via  penalized credible regions  with Dirichlet-Laplace  global-local shrinkage priors}
\end{center}
  \medskip
} \fi

\bigskip
\begin{abstract}
	The method of Bayesian variable selection via penalized credible regions  separates
	model fitting and variable selection. The idea is to  search for  the sparsest solution within the joint posterior credible regions. 
	Although the approach was successful, it depended on the use of conjugate normal priors. More recently, improvements in the use of global-local   shrinkage priors have been made for high-dimensional Bayesian variable selection. 
	In this paper,  we incorporate  global-local  priors into the credible region selection framework. 
	The  Dirichlet-Laplace (DL) prior    is adapted to linear regression. Posterior consistency for the normal and DL priors are shown, along with variable selection consistency. 
	We further introduce a new method to tune hyperparameters in prior distributions for linear regression. 
	We propose to choose the hyperparameters to minimize a discrepancy between the induced distribution on R-square and a prespecified target 
	distribution. Prior elicitation on R-square is more natural, particularly when there are a large number of predictor variables in which elicitation on that scale is not feasible. 
	For a normal prior, these hyperparameters are available in closed form to minimize the  Kullback-Leibler divergence between the distributions. 
\end{abstract}

\noindent%
\vfill

\newpage
\spacingset{1.45} 

\section{Introduction}\label{sec_introduction_credible}
High dimensional data has become increasingly common in all fields.  Linear regression is a standard and intuitive way to model dependency in high dimensional data. Consider the linear regression model: 
\begin{equation} 
	\label{linear_regression}
	\bm{Y} = \bm{X}\bm{\beta} +\bm{\varepsilon}
\end{equation}
where $\bm{X}$ is the $n\times p$  high-dimensional set of covariates, $\bm{Y}$ is the $n$ scalar responses,  $\bm\beta=(\beta_{1},\cdots,\beta_p)$ is the $p$-dimensional coefficient vector, and $\bm\varepsilon$ is the error term assumed to have $\text{E}(\bm\varepsilon) = 0$ and $\text{Var}(\bm\varepsilon) = \sigma^2 I_n$. 
Ordinary least squares is not  feasible   when the number of predictors $p$ is  larger than the sample size  $n$.  Variable selection is necessary to reduce the large number of candidate predictors.   
The classical variable selection methods include subset selection,  criteria such as AIC  \citep{akaike1998information}  and BIC \citep{schwarz1978estimating},  and penalized methods such as 
the least absolute shrinkage and selection operator  \citep[Lasso;][]{tibshirani1996regression}, smoothly clipped absolute deviation   \citep[SCAD;][]{fan2001variable}, the elastic net  \citep{zou2005regularization}, adaptive Lasso \citep{zou2006adaptive}, the Dantzig selector \citep{candes2007dantzig}, and octagonal shrinkage and clustering algorithm for regression   \citep[OSCAR;][]{bondell2008simultaneous}.

In the Bayesian framework, approaches for  variable selection include: stochastic search variable selection (SSVS)  \citep{george1993variable},    Bayesian regularization \citep{park2008bayesian,li2010bayesian,polson2013bayesian,leng2014bayesian},  empirical Bayes variable selection \citep{george2000calibration}, spike and slab variable selection \citep{ishwaran2005spike},    and global-local (GL) shrinkage priors.  Those traditional Bayesian methods conduct variable selection either  relying  on the calculation of posterior inclusion probabilities for each predictor or each possible model, or a choice of posterior threshold. 

Typical global-local shrinkage priors  are represented as the class of global-local scale mixtures of normals  \citep{polson2010shrink}, 
\begin{equation}\label{equation-globalLocal}
	\beta_j \sim N(0, w\xi_j),  \ \xi_j\sim \pi(\xi_j), \ (w, \sigma^2) \sim\pi(w, \sigma^2),
\end{equation}
where $w$ controls the global shrinkage towards the origin, while $\xi_j$ allows local deviations of shrinkage. 
Various  options  of shrinkage priors for $\bm\beta$, include 
normal-gamma \citep{griffin2010inference},    Horseshoe prior \citep{carvalho2009handling,carvalho2010horseshoe}, generalized double Pareto prior \citep{armagan2013generalized}, 
Dirichlet-Laplace (DL) prior   \citep{bhattacharya2015dirichlet}, 
Horseshoe+ prior \citep{bhadra2015horseshoe+}, 
and others  that  can  be represented as (\ref{equation-globalLocal}). 
The  GL shrinkage priors  usually shrink small coefficients greatly due to a   tight peak at zero, and rarely shrink large coefficients due to the heavy tails. 
It has been shown that GL shrinkage priors have improved posterior concentrations   \citep{bhattacharya2015dirichlet}.  
However, the shrinkage prior itself would not lead to variable selection, and to go further, some rules need to be set on the posteriors.

\cite{bondell2012consistent} proposed a Bayesian  variable selection method only based on   posterior credible regions.   
However, the implementation and results of that paper depended on  the use of conjugate normal priors. 
Due to the improved concentration, incorporating the global-local shrinkage priors into this framework can perform better,  both in theory and practice.
We show that  the DL prior yields consistent posteriors  in this regression setting, along with selection consistency. 

Another difficulty in high dimensional data  is the choice of hyperparameters, 
which can highly affect the   results.   In this paper, we also propose an intuitive default method  to tune the  hyperparameters in the prior distributions. 
By minimizing  a discrepancy between the induced  distribution of $R^2$ from the prior and  the desired distribution (Beta distribution  by default), one gets a default choice of hyperparameter value.   For the choice of normal priors, the hyperparameter  that minimizes the Kullback-Leibler (KL) divergence between the distributions is shown to have a closed form solution. 

Overall, compared to other Bayesian methods,  on the one hand, our method makes use of the advantage of global-local shrinkage priors, which can  effectively shrink small coefficients and reliably estimate the coefficients of important variables simultaneously. On the other hand, by using the credible region variable selection approach, we can  easily transform the non-sparse posterior estimators to sparse solutions. Compared to the common frequentist method, our approach   provides flexibility to estimate the tuning parameter jointly with the regression coefficients, allows easy incorporation
of external information or hierarchical modeling into Bayesian regularization framework,  and leads to  straightforward computing through Gibbs sampling.

The remainder of the paper is organized as follows. 
Section  \ref{sec_background} reviews   the penalized credible region variable selection method. 
Section  \ref{sec_methods} details the proposed method which combines  shrinkage priors and penalized credible region variable selection. 
Section  \ref{sec_asymptotics} presents the posterior consistency under  the choice of shrinkage priors, as well as the 
asymptotic behavior of the selection consistency for diverging $p$.  
Section  \ref{section-tuning hyperparameters} discusses a default method to tune the hyperparameters in the prior distributions based on the induced prior distribution on $R^2$. 
Section  \ref{sec_simulations} reports the simulation results, and Section  \ref{sec_real data} gives the analysis of a real-time PCR dataset. 
All proofs are given in the Appendix. 

\section{Background} \label{sec_background}

\cite{bondell2012consistent} proposed a penalized regression method based on Bayesian credible regions.  First, the full model is fit using all predictors with a   continuous prior.  
Then based on  the posterior distribution, a sequence of joint credible regions  are constructed, within which, one searches for the sparsest solution. 
The choice of  a conjugate normal prior of 
\begin{eqnarray} \label{normal-prior}
	\bm\beta|\sigma^2,\gamma\sim N(0,\sigma^2/\gamma \bm I_p) 
\end{eqnarray}
is used,  where $\sigma^2$ is the  error variance term  as in (\ref{linear_regression}), and $\gamma$ is the ratio of prior  precision to error precision. The variance,   $\sigma^2$, is often given a diffuse inverse Gamma prior, while $\gamma$ is the hyperparameter which is either chosen to be fixed or given a Gamma hyperprior.

The credible region is  to find $\tilde{\bm\beta}$,  such that 
\begin{equation}
	\label{2}
	\tilde{\bm\beta} =\argmin_{\bm\beta} ||\bm\beta||_0 \text{\ subject \ to \ } \bm\beta\in \mathcal{C}_\alpha,
\end{equation}
where $||\bm\beta||_0$ is  the $L_0$ norm of  $\bm\beta$, i.e., the number of nonzero elements, and $\mathcal{C}_\alpha$ is the $(1-\alpha)\times100 \%$  posterior credible regions based on the particular prior distributions.  
The use of elliptical posterior credible regions yields 
the form  $\mathcal{C}_\alpha = \{{\bm\beta}:({\bm\beta}-\hat{{\bm\beta}})^T\bm\Sigma^{-1}({\bm\beta}-\hat{{\bm\beta}})\leq c_\alpha\}$, for some nonnegative $c_\alpha$, where  $\hat{{\bm\beta}}$ and ${\bm\Sigma}$ are the posterior mean and covariance respectively. 
Then by replacing the $L_0$ penalization in (\ref{2}) with a smooth homotopy between $L_0$ and $L_1$ proposed by  \cite{lv2009unified} and  linear approximation, the optimization problem in (\ref{2}) becomes 
\begin{equation} 
	\label{eq_penalized_optimization}
	\tilde{{\bm\beta}} =\argmin_{\bm\beta}({\bm\beta}-\hat{{\bm\beta}})^T\bm\Sigma^{-1}({\bm\beta}-\hat{{\bm\beta}}) + \lambda_\alpha\sum_{j=1}^p |\hat\beta_j|^{-2}|\beta_j|,
\end{equation}
where there exists a one-to-one correspondence between $c_\alpha$ and $\lambda_\alpha$. The sequence of solutions
to (\ref{eq_penalized_optimization}) can be   directly accomplished  by  plugging in   the  posterior mean and covariance and using the  LARS algorithm \citep{efron2004least}.

\section{Penalized Credible Regions with Global-Local Shrinkage Priors} \label{sec_methods}

\subsection{Motivation}
Global-local shrinkage priors produce a posterior distribution with good empirical and theoretical properties.  
Compared to the usual normal prior,  GL priors concentrate more  along the regions with  zero parameters. 
This leads to a  better estimate of the uncertainty about parameters in the full model based on the posterior distribution. 
The penalized credible region variable selection approach  separates  model fitting and variable selection.  
So it seems natural  to fit the model under a   GL shrinkage prior, and then conduct variable selection through the penalized credible region method.  The motivation is that within the same credible region level, GL shrinkage priors would lead to more concentrated posteriors, thus having better performance for variable selection, by finding sparse solutions more easily. 
\cite{bondell2012consistent} demonstrated via simulations and real data examples that the credible region approach using the normal prior distribution improved on the performance of both Bayesian Stochastic Search and Frequentist approaches, such as Lasso, Dantzig Selector, and SCAD. The use of the GL shrinkage priors instead of the normal is a natural approach. In addition, although we do not have uncertainty about the model, the full posterior is obtained first, so that uncertainty about the parameters can be used based on the full model posterior distribution. Using the global-local shrinkage prior gives a more concentrated posterior even if we did not add the penalized credible region model selection step to choose an estimate of the model.

Although  GL shrinkage priors  would not lead to elliptical posterior distributions, valid credible regions can still be constructed using elliptical contours.  These would no longer be the high  density regions, but would  remain valid regions. Elliptical contours would also be  reasonable approximations to the high density regions,  at least around the largest mode.
Thus, the penalized credible region selection method can  be feasibly performed by plugging the posterior mean and covariance matrix into the optimization algorithm (\ref{eq_penalized_optimization}).  
So given any GL  prior,  once MCMC steps produce the posterior samples, the sample mean, $\hat{\bm\beta}$, and sample covariance, $\bm\Sigma$, would hence be obtained, then variable selection can be performed through the penalized credible region method.  In this paper, we modify  the  Dirichlet-Laplace (DL)   prior 
to implement in the regression setting. We also consider the 
Laplace  prior,  also referred as Bayesian Lasso,   described in \cite{park2008bayesian} and \cite{hans2010model},  as
\begin{eqnarray} \label{equation_Laplace prior}
	\beta_j\sim \text{DE}(\sigma/\lambda)   \ (j=1\cdots,p), 
\end{eqnarray} 
where 
$\lambda$ is the Lasso parameter, controlling the global shrinkage.  

\subsection{Dirichlet-Laplace Priors}
For the normal mean model, \cite{bhattacharya2015dirichlet} proposed a new class of Dirichlet-Laplace (DL) shrinkage priors, possessing the optimal posterior concentration property. 
We  construct the generalization of the  DL priors for the   linear regression model. 
The proposed hierarchical DL prior  is as follows:  for $j=1,\cdots,p$, 
\begin{eqnarray} \label{eq_DL-DE-form}
	\beta_j|\sigma, \phi_j,\tau &\sim & \text{DE}(\sigma\phi_j \tau),   \nonumber \\  
	(\phi_1,\cdots,\phi_p)&\sim & \text{Dir}(a,\cdots,a),  \\
	\tau &\sim & \text{Ga}(pa,1/2). \nonumber 
\end{eqnarray}
where 
$\text{DE}(b)$ denotes a zero mean   Laplace kernel with density $ f(y) =  (2b)^{-1}\exp\{-|y|/b\}$ for $y\in\mathbb{R}$,  
$ \text{Dir}(a,\cdots,a)$ is the Dirichlet distribution with   concentration vector $(a,\cdots,a)$, and Ga$(pa, 1/2)$ denotes a  Gamma distribution with shape $pa$ and rate $1/2$.  
Here,   small values of $a$ would lead  most of $(\phi_1, \cdots, \phi_p)$ to be close to  zero and only few of them nonzero; while large values  allow less singularity at zero, thus controlling the sparsity of regression coefficients.  The $\phi_j$'s are the local scales, allowing deviations in the degree of shrinkage. 
As pointed out in \cite{bhattacharya2015dirichlet}, $\tau$ controls global shrinkage towards the origin and  to some extent determines  the tail behaviors of the marginal distribution of $\beta_j$'s. 
We also assume a common prior on  the variance term 
$ \sigma^2$, $\text{IG} (a_1,b_1)$,    the inverse Gamma distribution with shape $a_1$ and scale $b_1$.

\subsection{Computation of Posteriors}
For posterior computation,  the Gibbs sampling steps proposed in \cite{bhattacharya2015dirichlet} can be 
modified to accommodate  the linear regression model.  
The DL prior (\ref{eq_DL-DE-form}) can be equivalently denoted as 
\begin{eqnarray} \label{equation-dl-normal-form}
	\beta_j|\sigma^2, \phi_j,\psi_j, \tau   &\sim & N(0,\sigma^2\psi_j\phi_j^2\tau^2), \nonumber \\ 
	\psi_j &\sim & \text{Exp}  (1/2),    \\ 
	(\phi_1,\cdots,\phi_p) &\sim &\text{Dir}(a,\cdots,a), \nonumber \\  
	\tau  &\sim &\text{Ga}(pa,1/2),  \nonumber
\end{eqnarray}    
where Exp$(\cdot)$ is the usual exponential distribution. 
Note that DL prior is also a  global-local shrinkage prior  as it is  a particular  form of  (\ref{equation-globalLocal}).  Gibbs sampling steps would be obtained based on (\ref{equation-dl-normal-form}). 
Since $\pi(\psi,\phi,\tau |\beta, \sigma^2) = \pi(\psi| \phi,\tau,\beta, \sigma^2)\pi(\tau|\phi,\beta, \sigma^2) \pi(\phi|\beta, \sigma^2)$, and the joint posterior of $(\psi,\phi,\tau)$ is  independent of $y$ conditionally on  $\beta$ and $\sigma^2$, so the steps to draw posteriors steps  are as follows:
(i) $\sigma^2 | \beta, \psi, \phi, \tau, y$, 
(ii)	$\beta| \psi,\phi,\tau, \sigma^2,y$,
(iii) $\psi|\phi,\tau,\beta, \sigma^2$,
(iv) $\tau|\phi, \beta, \sigma^2$,
(v) $\phi|\beta, \sigma^2$. 
The derivation is similar as  in \cite{bhattacharya2015dirichlet}, hence omitted here.

The parameterization of the three-parameter generalized inverse Gaussian (giG) distribution, $Y\sim\text{giG}(\chi,\rho,\lambda_0)$,   means the density of $Y$ is  $f(y) \propto y^{\lambda_0-1}\exp  \{-0.5(\rho y + \chi/y)\}$ for $y>0$. Then the summary of the Gibbs sampling steps are as below:
\begin{enumerate}[label=(\roman*)]
	\item  Sample $\sigma^2 | \beta,\psi,\phi,\tau,y$. Draw $\sigma^2 $ from an inverse Gamma distribution, IG$(a_1+ (n+p)/2, b_1 +( \bm\beta^T \bm {S}^{-1}\bm\beta  +(\bm{Y}-\bm{X}\bm\beta)^T  (\bm{Y}-\bm{X}\bm\beta) )/2 )$, where ${\bm S}= \text{diag} ( \psi_1\phi_1^2\tau^2, \cdots,\psi_p\phi_p^2\tau^2)$.

	\item Sample $\beta| \psi,\phi,\tau, \sigma^2, y$.  Draw $\bm\beta$  from a $N(\bm\mu,\sigma^2 {\bm V})$, where 
	$ \bm{V} = ( \bm{X}^T\bm{X} +{\bm S}^{-1})^{-1} $ with the same ${\bm S}$ as above, and $ \bm\mu =\bm{V} \bm{X}^T\bm{Y} = (\bm{X}^T\bm{X}+{\bm S}^{-1})^{-1}(\bm{X}^T\bm{Y})$.

	\item Sample $\psi_j| \phi_j,\tau,\beta, \sigma^2$. First draw $\psi_j^{-1}|\phi_j,\tau, \beta, \sigma^2$, $j=1,\cdots,p$, independently from the distribution  InvGaussian$(\mu_j = \sigma\phi_j\tau/|\beta_j|, \lambda_0 = 1) $, where InvGaussian$(\mu, \lambda_0)$ denotes the inverse Gaussian with density $f(y) = \sqrt{\lambda_0/(2\pi y^3)} \exp\{ -\lambda_0(y-\mu)^2 /(2\mu^2 y) \}$ for $y>0$. Then take the  reciprocal to get the draws of $\psi_j$ ($j=1,\cdots,p$).

	\item Sample  $\tau| \phi,\beta, \sigma^2$. Draw $\tau$ from a giG$(\chi=2\sum_{j=1}^p |\beta_j|/(\phi_j\sigma),\rho=1,\lambda_0=pa-p)$.

	\item  Sample $\phi_j|\beta, \sigma^2$. Draw $T_1,\cdots,T_p$ independently with $T_j\sim \text{giG}(\chi=2|\beta_j|/\sigma,\rho=1,\lambda_0=a-1)$, then set $\phi_j = T_j/T$ where $T=\sum_{j=1}^pT_j$.
\end{enumerate}

\section{Asymptotic Theory} \label{sec_asymptotics}
In this section,  we   first study the posterior properties of the   normal    and  DL prior,  when both $n$ and $p_n$ go  to infinity,   and further  investigate the selection consistency of the penalized variable selection method. 
Assume the true regression parameter is $\bm\beta_n^0$, and the estimated regression parameter is $\bm\beta_n$.
Denote  the  true set of non-zero coefficients is $\mathcal{A}_n^0=\{j:\beta_{nj}^0\neq0, j=1,\cdots, p_n\}$,
and the 
estimated set of non-zero coefficients is $\mathcal{A}_n=\{j:\beta_{nj}\neq0, j=1,\cdots, p_n\}$. 
Also let $q_n = |\mathcal{A}_n^0|$ denote  the number of predictors with nonzero true coefficients. 
As $n\rightarrow\infty$, consider the sequence of credible sets of the form  $\{ \bm\beta_n: (\bm\beta_n-\hat{\bm\beta}_n)^T\bm\Sigma^{-1}_n(\bm\beta_n-\hat{\bm\beta}_n)\leq c_n \}$, where $\hat{\bm\beta}_n$ and $\bm\Sigma_n$ are the posterior mean and covariance matrix respectively, and $c_n$ is  a sequence of non-negative constants.
Let $\bm\Gamma_n$ denote the $p_n\times p_n$ matrix whose columns are eigenvectors of $\bm{X}_n^T\bm{X}_n/n$ ordered by decreasing eigenvalues,  i.e., $d_1\geq d_2\geq\cdots\geq d_{p_n}\geq 0$. Then $\bm{X}_n^T\bm{X}_n/n = \bm\Gamma_n\bm D_n\bm\Gamma_n^{T}$ where $\bm D_n = \text{diag}\{d_1,\cdots ,d_{p_n}\}$.  

Assume the following  regularity conditions throughout. 
\begin{enumerate} [label=(A\arabic*)]  
	\item  \label{a_error} The error terms $\varepsilon_i$, $i=1,\cdots,n$,   are  independent and identically distributed (i.i.d.)  with mean zero and finite variance $\sigma^2$; 
	\item  \label{a_eigenvalues} 
	$0<d_{\min}<\liminf_{n\rightarrow\infty} d_{p_n} \leq \limsup_{n\rightarrow\infty}d_1 < d_{\max} <\infty$, where $d_{\min}$ and $d_{\max}$ are fixed; 
	\item \label{a_beta0} $\limsup\limits_{n\rightarrow\infty}\max\limits_{j=1,\cdots,p_n} |\beta_{nj}^0| <\infty$; 
	\item  \label{a_p=o(n)} $p_n=o(n/\log n)$; 
	\item \label{a_min}    $\sum_{j\in \mathcal{A}_n^0} | \beta_{nj}^{0} |^{-1} \leq C_0 \sqrt{n/p_n}$ and 
	$\sum_{j\in \mathcal{A}_n^0} | \beta_{nj}^{0} |^{-2} \leq C_1  {n/(p_n\sqrt{\log n})}$,  
	for some $C_0, C_1 >0$. 
\end{enumerate}
Assumption \ref{a_eigenvalues}  regarding the eigenvalues  bounded away from $0$ and $\infty$ is a necessary condition for estimation consistency in the Bayesian methods, and also for the consistency of Ordinary Least Squares in the case of growing dimension but with $p_n=o(n)$. 
This is akin to the condition in the fixed dimension of $\bm X^T\bm X/n$ converging to a positive definite matrix  (Assumption (A2) in \cite{bondell2012consistent}).  The basic intuition is that without a lower bound on the eigenvalue, there is an asymptotic singularity, which then leaves a linear combination of the regression parameters that is not identifiable, i.e, it would have a variance that was infinite, hence could not be consistent. 
The upper bound,  on the other hand, ensures that there is a proper covariance matrix for every $p_n$. 
If we assume that each row of $\bm X_n$ was a random draw from a $p_n$-dimensional probability distribution, the bounded eigenvalue condition is an assumption on the true sequence of covariance matrices, as for large $n$  and $p_n = o(n)$,  the sample covariance, $\bm X_n^T \bm X_n / n$, (assuming centered variables) will converge to the true covariance. Typical covariance structures will have the bounded eigenvalue property. 

Also note that Assumption \ref{a_min} restricts the minimum signal size for the non-zero coefficients while also ensuring that there are not too many small signals. 

\subsection{Posterior Consistency: Normal and DL Priors}
\cite{armagan2013posterior}  investigates the asymptotic behavior of posterior distributions of regression coefficients in the linear regression model (\ref{linear_regression}) as $p$ grows with $n$.
They prove the posterior consistency under the assumption of  a variety of  priors,
including the Laplace prior, Student's $t$ prior, generalized double Pareto prior, and the Horseshoe-like priors.  By definition,   posterior consistency implies that  the posterior distribution  of $\bm\beta_n$ converges in probability to  $\bm\beta_n^0$, i.e., 
for any $\epsilon>0$, $P(\bm\beta_n:||\bm\beta_n-\bm\beta_n^0||>\epsilon |\bm{Y}_n)\rightarrow 0$   as $ p_n, n\rightarrow \infty$.  
In this section,  we show that  the normal   and   Dirichlet-Laplace prior     also yield consistent posteriors.  However, the DL prior  can yield consistent posteriors under weaker conditions on the signal. 

\begin{theorem} \label{theorem_normal_consistency}
	Under assumptions  \ref{a_error}-\ref{a_beta0},  and  $p_n=o(n)$, if  $q_n=o\{n^{1-\rho}/(p_n(\log n)^2) \}$ for $\rho \in (0,1)$,    and  $\sqrt{\sigma^2/\gamma_n} = C/(\sqrt{p_n}n^{\rho/2}\log n)$  for finite $C>0$, 
	the normal prior  (\ref{normal-prior})
	yields a   consistent posterior. 
\end{theorem}

\begin{theorem} \label{theorem_DL_consistency}
	Under assumptions    \ref{a_error}-\ref{a_beta0}, and $p_n=o(n)$, if  $q_n = o( n/\log n)$, 
	and $a_n = C/( p_n n^{\rho}\log n)$ for any finite $\rho>0$ and finite $C>0$, 
	the Dirichlet-Laplace  prior  (\ref{eq_DL-DE-form}) yields a  consistent posterior.
\end{theorem}

Note that the   difference in the above two theorems  is the    number of nonzero components, i.e.,  $q_n$. As $ n /\log n  >n^{1-\rho}/(p_n(\log n)^2) $, 
the Dirichlet-Laplace prior leads to    posterior consistency  in a much broader domain, compared to the  normal prior as well as compared to  the Laplace prior who also yields consistent posteriors as shown in Theorem 2 in \cite{armagan2013posterior}.   
This   strengthens  the justification for  replacing the normal prior with the DL prior theoretically. 
However, note that the Theorems only give a sufficient condition for posterior consistency under each of the priors. The sufficient condition does have a broader domain for $q_n$ in Theorem 2, for the Dirichlet-Laplace prior, than in Theorem 1 for the normal prior. However, it is not clear that these conditions are also necessary, so although we are able to prove the consistency for the Dirichlet-Laplace prior under a more general condition than the normal prior, there may be room to improve this condition in either or both of these cases.

\subsection{Selection Consistency of Penalized Credible Regions}

\cite{bondell2012consistent} has shown that when  $p$ is fixed and $\bm\beta$ is given the  normal prior in (\ref{normal-prior}), the penalized credible region method is consistent in variable selection.   In this paper, we   show that the   consistency of the posterior distribution under  a global-local shrinkage prior also yields consistency in variable selection under the case of $p_n\rightarrow\infty$.

\begin{theorem}\label{theorem_selection consistency}
	Under assumptions \ref{a_error} - \ref{a_min},  given the normal prior in (\ref{normal-prior}),  if  $ c_n /(p_n\log n)\rightarrow  c$, where $c\in(0,\infty)$ is some constant value, 
	and  the prior precision,   $\gamma_n=o(n)$, then the  penalized credible region method is consistent in variable selection, i.e. $P(\mathcal{A}_n = \mathcal{A}_n^0)\rightarrow 1$. 
\end{theorem}

The proof is  given in  the Appendix.   
The selection consistency  allows us to expect that 
the true model is contained in the credible regions with high probability, when the number of predictors  increases together with  the sample size.    Such selection consistency is obtained under the normal prior.  However,  as reviewed in Section  \ref{sec_introduction_credible},   since the GL shrinkage priors can be expressed as a scale  mixture of normals,    as long as the 
posterior   distribution of the precision is $o(n)$  with probability 1 (analogous to $\gamma_n = o(n)$ in the normal prior), then the result can be directly applied to the GL shrinkage prior.   
\begin{theorem}\label{theorem_gl_selection }
	Under assumptions \ref{a_error} - \ref{a_min},   given any global-local shrinkage prior represented as (\ref{equation-globalLocal}), when  the conditions of posterior consistency are satisfied,   then the  posterior   distribution of the precision is $o(n)$  with probability 1 as $n\rightarrow\infty$.
	Furthermore, if  $ c_n /(p_n\log n)\rightarrow  c$, where $c\in(0,\infty)$ is some constant value,  then the  penalized credible region method with the particular shrinkage prior is consistent in variable selection, i.e. $P(\mathcal{A}_n = \mathcal{A}_n^0)\rightarrow 1$. 	
\end{theorem}
So given  the conditions of posterior consistency under  the global-local shrinkage prior, we automatically get   the selection consistency of the credible region method. For example, for the DL prior  in (\ref{eq_DL-DE-form}),   we have the following result. 
\begin{corollary}
	Under assumptions \ref{a_error} - \ref{a_min},      given the DL prior in (\ref{eq_DL-DE-form}), if  $q_n = o( n/\log n)$, 
	$a_n = C/( p_n n^{\rho}\log n)$ for any finite $\rho>0$ and finite $C>0$,  
	if  $ c_n /(p_n\log n)\rightarrow  c$, where $c\in(0,\infty)$ is some constant value,   then the  penalized credible region method is consistent in variable selection, i.e. $P(\mathcal{A}_n = \mathcal{A}_n^0)\rightarrow 1$. 
\end{corollary}

Note  that the variable selection consistency is derived based on the posterior consistency. However, Assumption \ref{a_beta0} is not necessary to ensure the variable selection consistency.   If \ref{a_beta0} is not satisfied, i.e., 
$\bm\beta_n^0$ is truly unbounded, although it would not be possible to obtain a consistent estimator, or posterior, the credible region would become bounded away from zero in that direction, and hence  will pick out that direction consistently as well.

\section{Tuning  Hyperparameters} \label{section-tuning hyperparameters}
The value   of  hyperparameters in the prior distribution  plays an important role in the posteriors. 
For example, in the normal prior (\ref{normal-prior}),  $\gamma$ is the  hyperparameter, whose value controls the degree of shrinkage.
This is often chosen to be fixed at a ``large'' value  or given a hyperprior. However, the choice of the  ``large'' value affects the results, as does the choice of hyperprior such as a gamma prior, particularly in the high dimensional case. 
Also,  in the   DL prior (\ref{eq_DL-DE-form}), the choice of   $a$ is critical.  If $a$ is too small, then the DL prior would shrink  each dimension of  $\bm\beta$ towards zero; while, if $a$ is too large, there would be no strong concentration around the origin. 
Instead of fixing $a$, a discrete uniform prior can be given on $a$ supported on some interval (for example, $ [ 1 /\text{max}(n,p) , 1/2]$), with several support points on the interval.  
However, introducing the hyperprior for the hyperparameters  will not only arise new   values to tune, but also increase the complexity of the MCMC sampling.    
In practice, although the specification of a $p$-dimensional prior on $\bm\beta$ may be difficult, some prior information on a univariate function may be easier. The motivation is to incorporate such prior information of the one-dimensional function into the priors on the $p$-dimensional $\bm\beta$.

In this paper, we propose an intuitive way to tune the values of hyperparameters, by incorporating a prior on $R^2$ (the coefficient of determination). Practically, a scientist may have information on $R^2$ from previous experiments, and this can be coerced into say a Beta$(a, b)$ distribution. In this way, tuning hyperparameters is equivalent to searching for the hyperparameter which leads to the induced distribution of $R^2$ closest to the desired distribution. Intuitively, if we fix any value for $b$, as we increase $a$,   then $R^2$ will approach 1, hence this controls the total size of the signal that is anticipated in the data. As we will see shortly, it is a prior on the value of the quadratic form $\bm\beta^T \bm X^T \bm X \bm \beta$. Combining this with the choice of prior, gives also the degree of sparsity. For example, with a Dirichlet-Laplace Prior, the parameter in the DL distribution then controls how this total signal is distributed to the coefficients, either to a few coefficients, giving a sparse model, or to many coefficients, giving a dense model.  In many cases, a scientist may have done many similar experiments before and can look back and see the values of the sample coefficient of determination from all of these studies. Then treating this as a sample from a Beta distribution, the hyperparameters, $ a$  and $b$, can be obtained from this fit. Without any prior information for $R^2$, a uniform prior, Beta$(1, 1)$, may be used as default.


For  the linear regression model (\ref{linear_regression}), the population form can be represented as 
$y=\bm{x}^T\bm\beta + \varepsilon$, with $\bm{x}$ independent  of  $\varepsilon$. 
Let  $\sigma_y^2$ be the 
marginal  variance of $y$  and $\sigma^2$  be the   variance of  the random error term. The definition of the  POPULATION $R^2$ is given by:  
\begin{eqnarray}
	\label{equation-pop-R-square}
	pop \ R^2 = 1 - \frac{\sigma^2}{\sigma^2_y},  \nonumber
\end{eqnarray}
which is the  proportion of the variation of $y$ in the population explained by the independent variables. 
Furthermore, for fixed $\bm\beta$, it follows that $\sigma_y^2 = \bm\beta^T\text{Cov}(\bm{x})\bm\beta + \sigma^2$. Assume $E(\bm{x}) = 0$, then we can estimate $\text{Cov}(\bm{x})$  by  $\bm{X}^T\bm{X}/n$. 
So   $R^2$ as a function of $\bm\beta$ and $\sigma^2$ is given by 
$
R^2 = 1 -  \sigma^2 /(\bm\beta^T\bm{X}^T\bm{X}\bm\beta/n + \sigma^2).   
$
Given that the form of  prior distributions  considered includes $\sigma$ in the scale, 
it follows that $\bm\beta = \sigma\bm\eta$  for  $\bm\eta$ having the distribution of the prior  fixed with $\sigma^2=1$. Hence  
\begin{eqnarray} \label{equation-Rsquare}
	R^2 = 1 - \frac{1}{1+\bm\eta^T\bm{X}^T\bm{X}\bm\eta/n }. 
\end{eqnarray} 
For a specified prior  on  $\bm\eta$, the induced distribution of $R^2$ can be derived based on (\ref{equation-Rsquare}). 
Then the hyperparameters which yield the  induced distribution of $R^2$ closest to the desired distribution  is the tuned value. 

\begin{figure}[h!]
	\centering
	\includegraphics[scale=0.7]{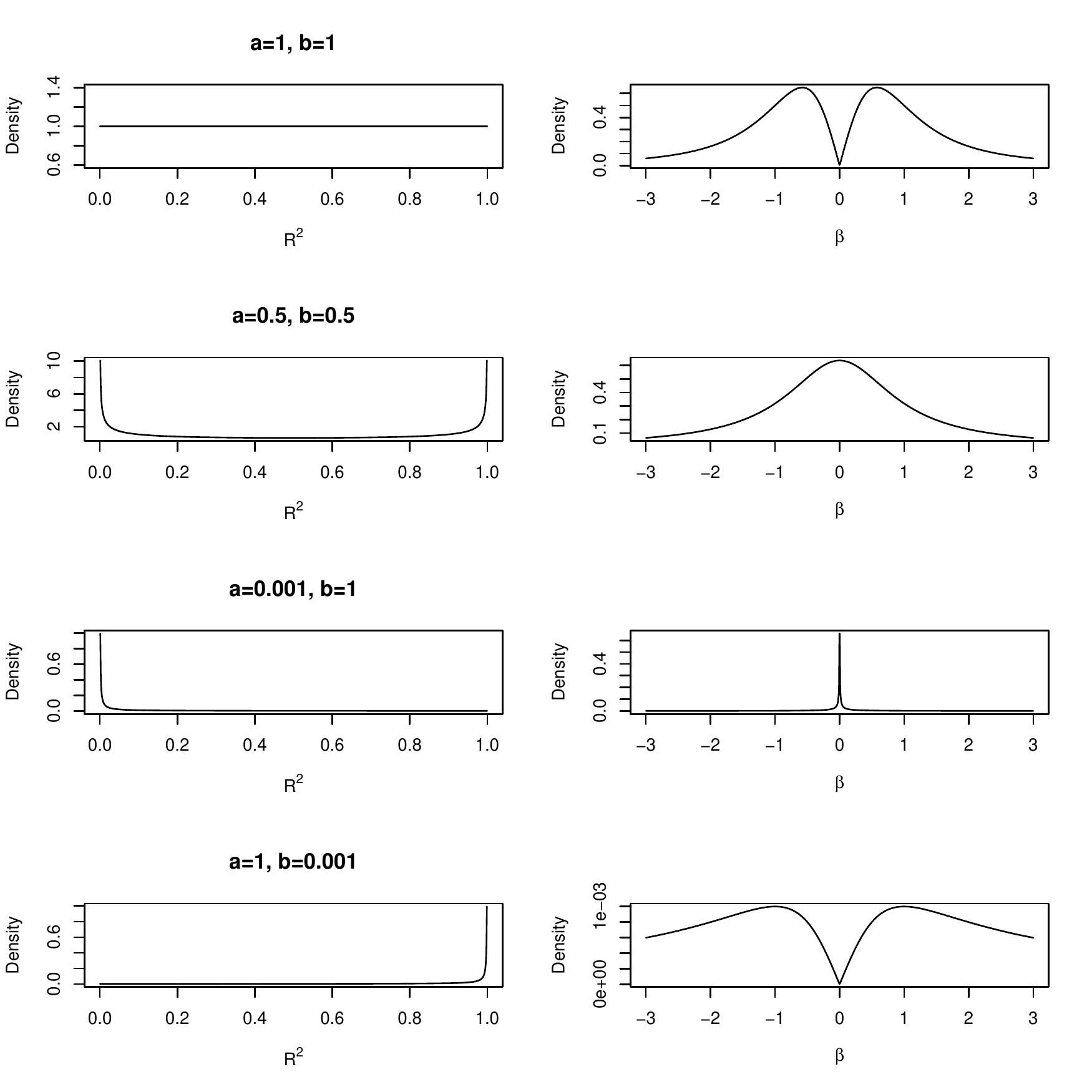}
	\caption{ Beta$(a,b)$ density for  $R^2$ and the corresponding induced distribution density for $\beta$.  }\label{rsquare-beta}
\end{figure}
For a better understanding of the intuition here,   we give a simple example. Suppose  $\sigma^2 =1$  and we have an intercept only model,  i.e., model  (\ref{linear_regression}) is simplified as $\bm{Y}=  \bf{1}_n\beta + \bm{\varepsilon}$ with $\bf 1_n$  the $n$-dimensional vector with all elements of 1. 
Then  (\ref{equation-Rsquare}) can be written as $R^2 = 1 - \frac{1}{1+\beta^2}$.  
Suppose the desired distribution for $R^2$ is Beta$(a,b)$, then the corresponding induced distribution for $\beta$ is  
\[ 
f_\beta(t) = \frac{2\Gamma(a+b)}{\Gamma(a)\Gamma(b)} (\frac{t^2}{1+t^2})^{a-1} (\frac{1}{1+t^2})^{b+1} |  t|, 
\]
where $\Gamma(.)$ denotes the gamma function. 
The left panel of Figure \ref{rsquare-beta}  shows the distribution on $R^2$  for 4 choices of hyperparameters in the Beta distribution, while the right panel shows the corresponding induced prior distribution on $\beta$. We see that for a uniform distribution on $R^2$, we obtain a distribution on $\beta$ that puts its mass slightly skewed away from zero on each side. For a bathtub distribution $(a = b = 0.5)$, we see it reduces to the Cauchy distribution, giving heavy tails to obtain the $R^2$ near one, and the peak around zero to obtain the $R^2$  near zero. We also see two other extremes, as $a \to 0 $ for fixed $b = 1$,  we obtain a distribution that decays very quickly and puts most of its mass around zero, as expected;  while as $b \to 0$ and $a$ fixed at 1, we obtain a density proportional to $| t | / (1 + t^2)$, allowing for larger values of $\beta$ with high probability.

In practice, one can consider    a grid of possible values of  the hyperparameters. For each value, draw a vector $ {\bm\eta}$. This is converted to a draw of $R^2$. Given this hyperparameter, a comparison between the sample of $R^2$ and the desired distribution is performed, for example, a Kolmogorov-Smirnov(KS) test. The best fit is then chosen. 
The whole tuning process only involves    the prior distributions, no    MCMC  sampling, thus avoiding comprehensive computing.

However,  given a specific prior for  $\bm\beta$,   based on (\ref{equation-Rsquare}),  the exact    induced distribution of $R^2$ can be derived,  which relies on the value of hyperparameters.    By minimizing  the Kullback-Liebler  directed divergence between such distribution and the desired  distribution (Beta distribution by default),  a default hyperparameter value can be  found. 
For continuous random variables with density function  $f_1$ and $f_2$, the KL divergence is defined as 
\[
D(f_1 | f_2)   =   \int_{-\infty}^\infty f_{1}(x) \log(f_{1}(x)/f_2(x))\, dx. 
\]      For the choice of normal priors, the following theorem shows that there is a closed form solution for the hyperparameter to minimize the KL divergence for large $p$. 
\begin{theorem}\label{theorem_normal tuning}
	For the normal prior in (\ref{normal-prior}),   to minimize the KL directed divergence between the induced distribution of $R^2$ and the Beta$(a,b)$ distribution,  as $p\rightarrow\infty$,    the hyperparameter,  $\gamma$,   is   chosen to be   $(A+\sqrt{B})^{1/3} + (A- \sqrt{B})^{1/3}  - P/3$, where   
	$
	P = (2a-b)\sum_{j=1}^p d_j /a$,  
	$Q = 2(a+b)  \sum_{j=1}^p d_j^2/a +  (a-2b)  (\sum_{j=1}^p d_j)^2/a$,   
	$R = -b( \sum_{j=1}^p d_j)^3/{a}$,  
	$
	C = P^2/9 - Q/3, \ 
	A = PQ/6 - P^3/27 - R/2$,   
	$B = A^2 - C^3\geq 0
	$, and   $d_1,\cdots ,d_{p}$  denote the eigenvalues of  $\bm{X}^T\bm{X}/n$. 
\end{theorem}

In theory, for   other continuous priors, one can derive the  optimal hyperparameters similarly.  However, sometimes the calculation can be quite complex. In this case, the simulation-based approach discussed earlier   can be implemented.
However, since   GL priors   can be represented as  mixture normal priors (see Section  \ref{sec_introduction_credible}),  by matching its prior precision   with that of the normal prior,  the   derived  default solution  as shown in Theorem \ref{theorem_normal tuning}  can offer  an intuitive idea for  the hyperparameter values in  the    GL shrinkage priors.

\section{Simulation Results} \label{sec_simulations}
\subsection{Comparisons of Different Priors} \label{sec_simu1}
To compare the performance of the penalized credible region variable selection method using   different  shrinkage priors, including the normal prior (\ref{normal-prior}),  Laplace  prior (\ref{equation_Laplace prior}), and DL prior (\ref{eq_DL-DE-form}), a simulation study is conducted.  
\cite{bondell2012consistent} demonstrated  the improvement in performance of the credible region approach using the normal prior over both Bayesian and Frequentist approaches,  such as SSVS, Lasso, adaptive Lasso,  Dantzig Selector,  and SCAD.  Given the previous comparisons, the focus here is to see if replacing the normal prior with the global-local prior can even further improve the performance of the credible region variable selection approach.

We use a similar simulation setup as in \cite{bondell2012consistent}. In each setting, 200 datasets are simulated from the linear model  (\ref{linear_regression}) with $\sigma^2=1$, sample size  $n=60$, and the number of predictors $p$  varying  in $ \{50, 500, 1000\}$. 
To represent different correlation settings,  $X_{ij}$ are generated from standard normal distribution, and the correlation between $x_{ij_1}$ and $x_{ij_2}$ is  $\rho^{|j_1-j_2|}$, with  $\rho=  0.5$ and $0.9$.  
The true coefficient  $\bm\beta$ is  $({\bm 0}_{10}^T,{\bm B_1}^T,{\bm0}_{20}^T,{\bm B_2}^T,{\bm 0}^T_{p-40})^T$ for $p\in \{50, 500, 1000\}$ in which  ${\bm 0}_k$ represents the $k$-dimensional zero vector, 
${\bm B_1}$ and $\bm B_2$ are both   $5$-dimensional vector   generated component-wise   and  uniform from $(0,1)$.  
For each case of shrinkage prior,  the  posterior mean and covariance can be obtained from the  Gibbs samplers, and then   plugged into the optimization algorithm  (\ref{eq_penalized_optimization}) of the penalized credible region method to implement the variable selection.

For each method,   the  induced ordering of the  predictors are created. We consider the resulting model at each ordering step  to measure the performance. 
For each step on the ordering,  true positives (TP) are defined as  those selected  variables which also appear in the true model. 
False positives (FP) are those selected variables  which also do not appear  in the true model. 
True negatives (TN) correspond to those not selected variables  which are   not   in the true model. 
False negatives (FN) refer to variables which are not selected in the model, but indeed are  in the true model. 
The Receiver-Operating Characteristic (ROC) curve plots the false positive rate (FPR or 1-Specificity) on the x-axis and the true positive rate (TPR or Sensitivity) on the y-axis, where FPR is the fraction of FP's of the fitted model in the  total number of irrelevant variables in the true model, and TPR is the fraction of TP's of the fitted model in the total number of important variables in the true model. 
The Precision-Recall (PRC) curve plots the precision on the y-axis, and the Recall (or TPR or Sensitivity) on x-axis, where precision is the ratio of true positives to the total  declared   positive number.   

The compared  credible set methods are listed as below: 
\begin{itemize}
	\item Method ``Normal\_hyper'', refers to the normal prior, with ``non-informative'' hyperparameters, i.e.,  $N(0,\sigma^2_b)$ is the prior for $ \beta$, and IG$(0.001,0.001)$  prior is given for  $\sigma^2_b$. 
	\item  Method ``Normal\_tune'', refers to the normal prior (\ref{normal-prior}), where $\gamma$ is  tuned through the  $R^2$ method introduced in Section  \ref{section-tuning hyperparameters}, with a target of uniform distribution. 
	%
	\item Method ``Laplace\_hyper'',  means Laplace prior (\ref{equation_Laplace prior}),  with   $\lambda$ given a $\text{Ga}(1,1)$ prior. 
	\item Method ``Laplace\_tune'',  means Laplace prior (\ref{equation_Laplace prior}),    and  $\lambda$ is tuned through the  $R^2$ method introduced in Section  \ref{section-tuning hyperparameters}, with a target of uniform distribution. 
	\item  Method  ``DL\_hyper'' is the DL prior (\ref{eq_DL-DE-form}), in which  $a$ is given a discrete uniform prior supported on the interval $[{1}/{\max(n,p)}, {1}/{2}]$ with 1000 support points in this interval.  
	\item Method ``DL\_tune'' is the DL prior (\ref{eq_DL-DE-form}), in which  
	$a$ is tuned through the  $R^2$ method introduced in Section  \ref{section-tuning hyperparameters}, with a target of uniform distribution. 
\end{itemize}
In all above cases, the variance term $\sigma^2$ is given an IG$(0.001, 0.001)$ prior. In addition, we show the results from using the  Lasso  \citep{tibshirani1996regression}  fit via the LARS algorithm  \citep{efron2004least}.  


\begin{table} [h!]
	\caption{Mean area under the ROC Curve and the PRC curve for $p=50$, $n=60$, based on 200 datasets with standard errors in parentheses. }
	\label{table50}
	\renewcommand{\arraystretch}{1.0}
	\begin{center}  
		\begin{tabular}{ l|      c  c     |     c c  }
			\hline \hline 
			& \multicolumn{2}{c|}{ROC Area} & \multicolumn{2}{c}{PRC Area} \\
			& $\rho=0.5$ & $\rho=0.9$ &  $\rho=0.5$ & $\rho=0.9$ \\
			\hline
			{\footnotesize Lasso   }   & {\footnotesize0.900  (0.0047)} 
			&{\footnotesize0.815 (0.0052)}  
			& {\footnotesize0.694  (0.0053)} & {\footnotesize0.628 (0.0068)} \\ \hline
			{\footnotesize Normal\_hyper }   & {\footnotesize0.909 (0.0048)}  
			& {\footnotesize0.899 (0.0041)}  
			& {\footnotesize0.782  (0.0054)} & {\footnotesize0.749 (0.0058)} \\
			{\footnotesize Normal\_tune}    & {\footnotesize0.949 (0.0037)} 
			&{\footnotesize0.978 (0.0020)}  
			& {\footnotesize0.830 (0.0043)} & {\footnotesize0.845 (0.0039)} \\      \hline
			{\footnotesize Laplace\_hyper }   & {\footnotesize0.890  (0.0049)}  
			&{\footnotesize0.859  (0.0052)}  
			& {\footnotesize0.756 (0.0058)} & {\footnotesize0.691  (0.0069)} \\
			{\footnotesize Laplace\_tune}  & {\footnotesize0.942 (0.0040)} 
			& {\footnotesize0.976 (0.0020)}  
			& {\footnotesize0.820 (0.0046)} & {\footnotesize0.844 (0.0039)}  \\ 
			\hline
			{\footnotesize DL\_hyper }    & {\footnotesize0.917 (0.0044)} 
			& {\footnotesize0.908 (0.0044)}  
			& {\footnotesize0.786  (0.0052)} & {\footnotesize0.749  (0.0062)}\\
			{\footnotesize DL\_tune} & {\footnotesize0.939  (0.0039)} 
			& {\footnotesize0.945  (0.0032)}  
			& {\footnotesize0.811  (0.0048)} & {\footnotesize0.802 (0.0050)}\\
			\hline 
		\end{tabular}
	\end{center}
\end{table}
%
\begin{figure}[h!]
	\centering
	\includegraphics[scale=0.8]{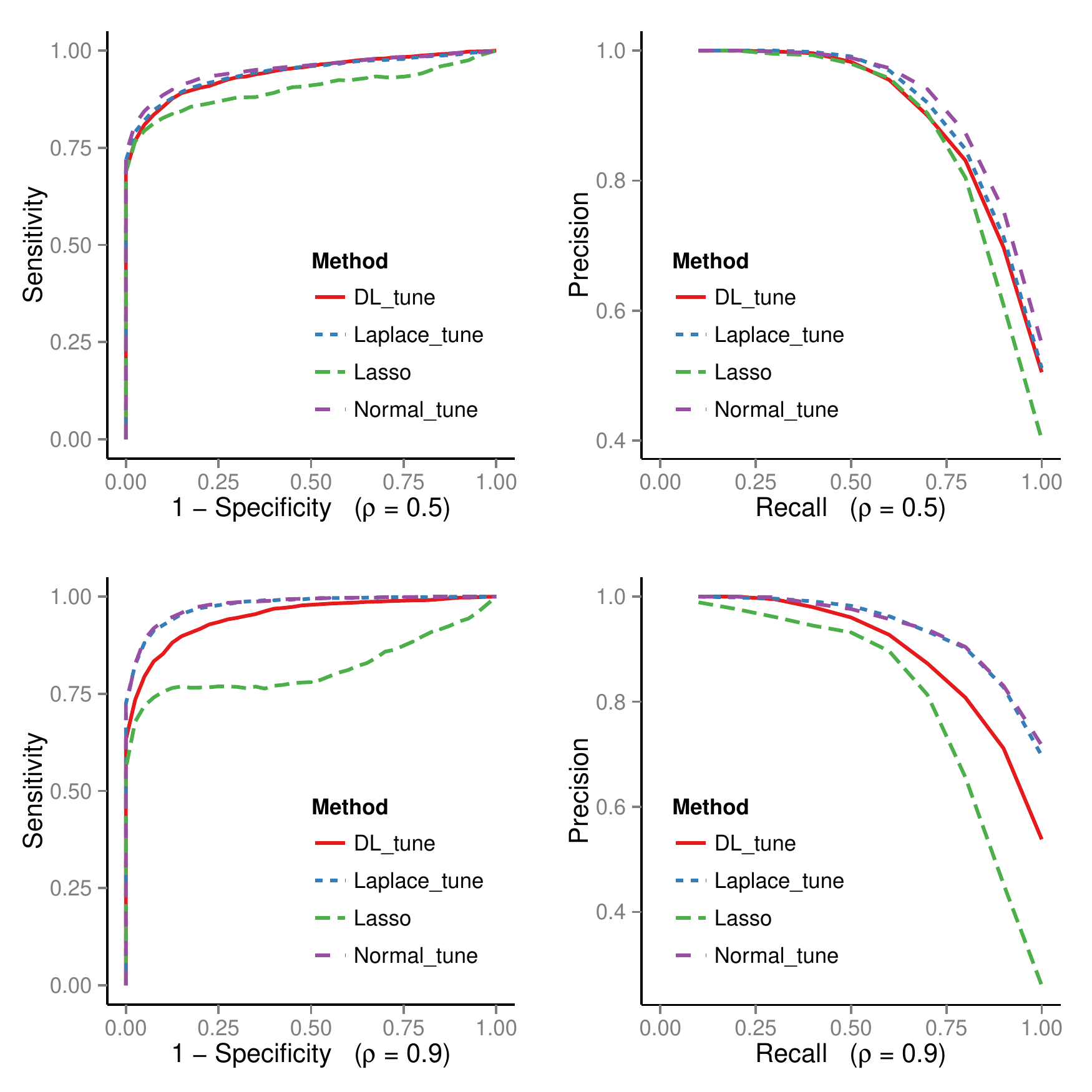}
	\caption{ Plot of mean ROC and PRC curves when $\rho=0.5$ and $\rho=0.9$,  over the 200 datasets for $p=50$ predictors, $n=60$ observations. The left column is the ROC curve, the right column is the PRC curve. }\label{figure50}
\end{figure}

For the above  priors (normal, Laplace and DL),  we ran the MCMC chain (Gibbs sampling) for $15,000$ iterations, with the first $5,000$ for burn-in.  Posterior mean and covariance were calculated based on the $10,000$ samples, which  were then  plugged into the penalized credible interval optimization algorithm (\ref{eq_penalized_optimization}), to conduct  variable selection. 
Table \ref{table50} gives the mean and standard error for the area under the ROC and PRC curve for $p=50$ with   $\rho\in\{0.5, 0.9\}$. 
In addition,  Figure \ref{figure50} plots the mean ROC and PRC curves  of the 200 datasets for  the selected  above methods to compare. 
Table \ref{table500} and Figure \ref{figure500} give the results for the $p=500$ case.
Table \ref{table1000} and Figure \ref{figure1000}  show the results for the $p=1000$ case.  Since  the Lasso  estimator can  select at most $\min\{n,p\}$ predictors,  when  $p=500$ or  $1000$,  the ROC and PRC curves cannot be fully constructed.  So the area under the curves cannot be compared directly for Lasso with other methods, which are omitted in  Table \ref{table500} and \ref{table1000}, but partial ROC and PRC curves can  still be plotted, which are shown in Figure \ref{figure500} and  \ref{figure1000}.  


\begin{table} [h!]
	\caption{Mean area under the ROC Curve and the PRC curve for $p=500$, $n=60$, based on 200 datasets with standard errors in parentheses. }\label{table500}
	\renewcommand{\arraystretch}{1.0}
	\begin{center}  
		\begin{tabular}{ l|      c  c     |     c c  }
			\hline \hline 
			& \multicolumn{2}{c|}{ROC Area} & \multicolumn{2}{c}{PRC Area} \\
			& $\rho=0.5$ & $\rho=0.9$ &  $\rho=0.5$ & $\rho=0.9$ \\
			\hline
			{\footnotesize Lasso}       & -
			& -  
			& {\footnotesize 0.550 (0.0087)} &{\footnotesize  0.550   (0.0089)} \\ \hline
			{\footnotesize Normal\_hyper}    & \footnotesize{0.948  (0.0031)} 
			& \footnotesize{0.990 (0.0013)}  
			& \footnotesize{0.615 (0.0093)} & \footnotesize{0.784  (0.0062)} \\     
			{\footnotesize Normal\_tune  }&  \footnotesize{0.950  (0.0029)} 
			& \footnotesize{0.992  (0.0007)}  
			& \footnotesize{0.610 (0.0091)} &\footnotesize{0.721  (0.0077)} \\  
			\hline   
			{\footnotesize Laplace\_hyper}   & \footnotesize{0.937  (0.0030)} 
			& \footnotesize{0.969  (0.0020)} 
			& \footnotesize{0.621  (0.0087)} &\footnotesize{0.680  (0.0087)}  \\
			{\footnotesize Laplace\_tune}   & {\footnotesize 0.959 (0.0027)} 
			& {\footnotesize 0.995 (0.0004)}  
			& {\footnotesize 0.701  (0.0077)} & {\footnotesize 0.822 (0.0055)}  \\  
			\hline  
			{\footnotesize DL\_hyper}     & {\footnotesize 0.927   (0.0038)} 
			& {\footnotesize 0.908   (0.0047)}  
			&{\footnotesize  0.651   (0.0092)} & {\footnotesize 0.570   (0.0102)}\\
			{\footnotesize DL\_tune}   & {\footnotesize 0.949   (0.0027)} 
			& {\footnotesize 0.970   (0.0025)}  
			& {\footnotesize 0.717   (0.0085)} & {\footnotesize 0.797   (0.0073)}\\
			\hline 
			%
		\end{tabular}
	\end{center}
\end{table}

\begin{figure}[h!]
	\centering
	\includegraphics[scale=0.8]{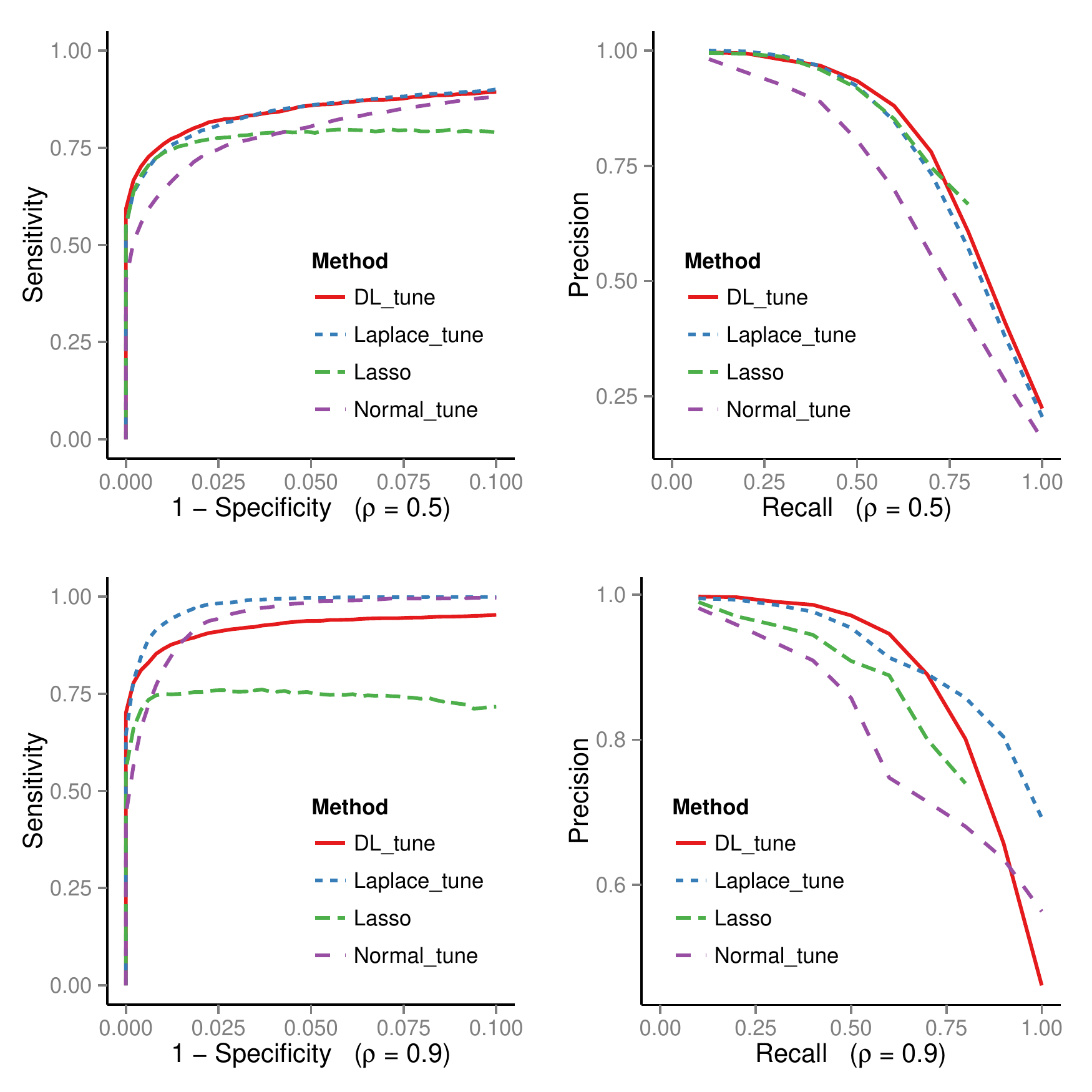}
	\caption{ Plot of mean ROC and PRC curves when $\rho=0.5$ and $\rho=0.9$,  over the 200 datasets for $p=500$ predictors, $n=60$ observations. The left column is the ROC curve, the right column is the PRC curve.}\label{figure500}
\end{figure}
%

\begin{table} [h!]
	\caption{Mean area under the ROC Curve and the PRC curve for $p=1000$, $n=60$, based on 200 datasets with standard errors in parentheses. }\label{table1000}
	\renewcommand{\arraystretch}{1.0}
	\begin{center}  
		\begin{tabular}{ l|      c  c     |     c c  }
			\hline \hline 
			& \multicolumn{2}{c|}{ROC Area} & \multicolumn{2}{c}{PRC Area} \\
			& $\rho=0.5$ & $\rho=0.9$ &  $\rho=0.5$ & $\rho=0.9$ \\
			\hline
			{\footnotesize Lasso  }   & -  
			& -  
			& {\footnotesize 0.507  (0.0093)} &  {\footnotesize0.536 (0.0091)} \\ \hline
			{\footnotesize Normal\_hyper}    &  {\footnotesize0.942 (0.0039)} 
			&  {\footnotesize0.992 (0.0018)} 
			& {\footnotesize0.515 (0.0101)} &  {\footnotesize0.727 (0.0076)} \\
			{\footnotesize Normal\_tune}  &  {\footnotesize0.943 (0.0039)} 
			&  {\footnotesize0.991 (0.0018)}  
			&  {\footnotesize0.539 (0.0099)} &  {\footnotesize0.680 (0.0083)} \\     
			\hline
			{\footnotesize Laplace\_hyper}   &  {\footnotesize0.914 (0.0041)} 
			&  {\footnotesize0.968 (0.0021)}  
			&  {\footnotesize0.444 (0.0093)} &  {\footnotesize0.554 (0.0102)}  \\
			{\footnotesize Laplace\_tune} &    {\footnotesize 0.951 (0.0038)} 
			&  {\footnotesize0.994 (0.0012)} 
			&  {\footnotesize0.638 (0.0092)} &  {\footnotesize0.764 (0.0071)}  \\   
			\hline
			%
			{\footnotesize DL\_hyper}  &  {\footnotesize0.931 (0.0040)} 
			&  {\footnotesize 0.943 (0.0034)}  
			&  {\footnotesize 0.635  (0.0096)} &  {\footnotesize0.623  (0.0094)}\\
			{\footnotesize DL\_tune}   &  {\footnotesize0.925  (0.0045)} 
			&  {\footnotesize0.967  (0.0025)} 
			&  {\footnotesize0.633  (0.0116)} &  {\footnotesize0.768  (0.0092)}\\
			\hline
			%
		\end{tabular}
	\end{center}
\end{table}
%
\begin{figure}[h!]
	\centering
	\includegraphics[scale=0.8]{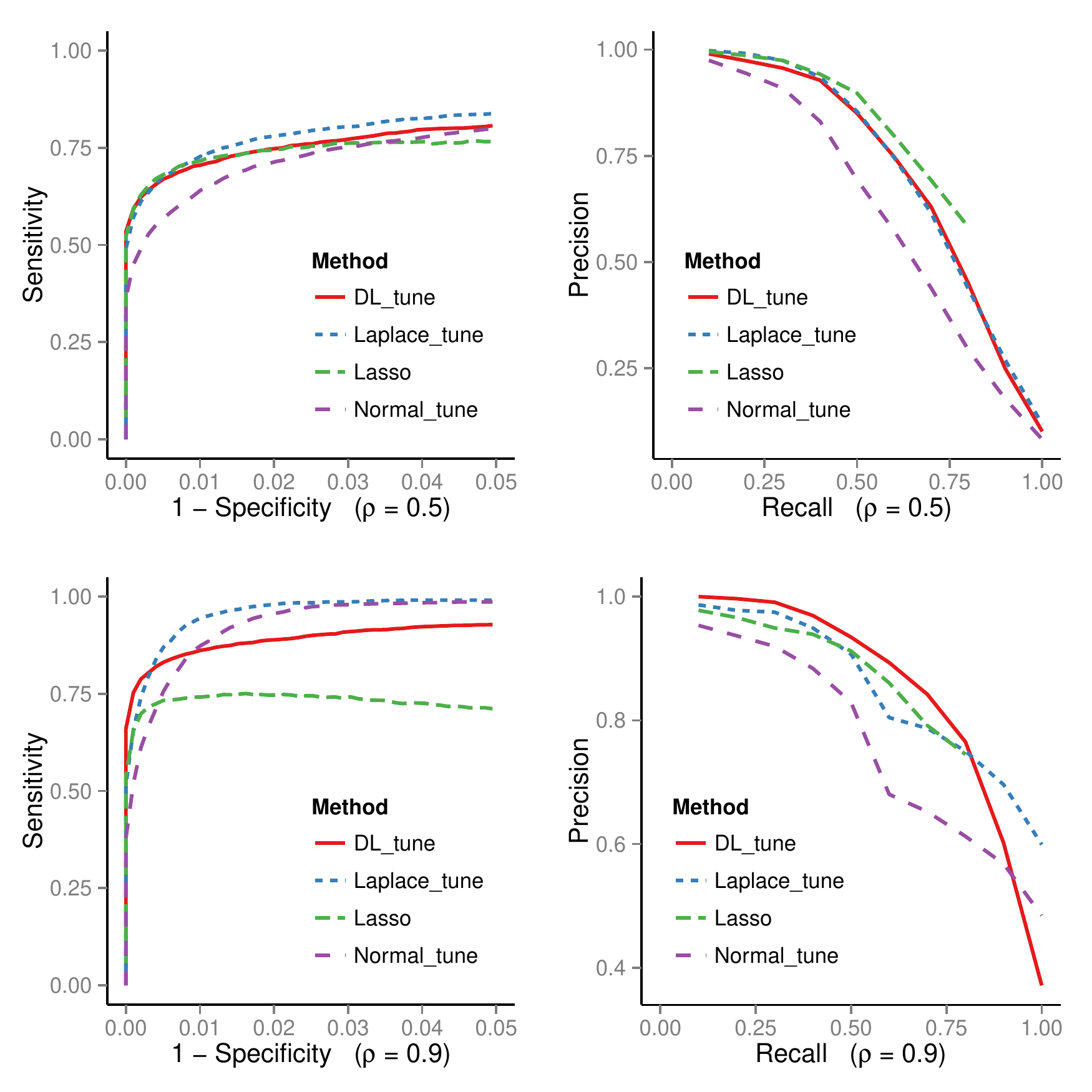}
	\caption{  Plot of mean ROC and PRC curves when $\rho=0.5$ and $\rho=0.9$,  over the 200 datasets for $p=1000$ predictors, $n=60$ observations. The left column is the ROC curve, the right column is the PRC curve.}\label{figure1000}
\end{figure}

On the one hand, in terms of whether given a  hyperprior for the hyperparameter  or tuning hyperparameters through the  $R^2$  method proposed in Section  \ref{section-tuning hyperparameters} would lead to    better posterior performance,  one might compare  each ``$\ast$\_hyper'' and ``$\ast$\_tune'' pair in  Table \ref{table50}, \ref{table500} and \ref{table1000}. 
In general, for all three priors, the  tuning method  leads to significantly better posterior performance than the hyperprior method in  all   simulation setups.

On the other hand,   in terms of comparing performance of  different   priors applied on the penalized credible region variable selection,  combining both the tables and figures, we have the following findings.  
When considering the Precision-Recall in particular, the DL and Laplace priors outperform the normal prior and Lasso. This is particularly true if the hyperparameters in them are tuned via a uniform distribution on $R^2$.  We note that when there are only a few true   and many unimportant variables, the Precision-Recall curve is a more appropriate measure than the  ROC curve. For example,  when $p=1000$, in both $\rho=0.5$ and $0.9$ cases,  
in Figure \ref{figure1000}, the PRC curve  shows that the DL prior is significantly better than the normal prior;   the ROC curve of the normal prior goes higher when FPR (or 1-Specificity) is large, however, when  FPR is small (which is  of more  interest),  DL prior still leads to significantly larger sensitivity than the normal prior. Overall, the DL prior outperforms the normal prior,   as does the Laplace prior.

\subsection{Additional Simulations on  Hyperparameter Tuning}
To examine the role of $a$ in the DL prior,  additional simulations were conducted. Table \ref{table-dl2np}  gives the average squared error for the posterior mean based on the 200  same  datasets as Section  \ref{sec_simu1}, for the  DL priors with $a$ fixed at $1/2$, $1/n$, and $1/p$.  
The results show that when  $p$ is large or there is strong correlation in the dataset,  $a=1/n$  is better than $a=1/2$. 
When $p$ is small and there is only moderate correlation for the data,  $a=1/2$  is recommended. 
%
Since the performance of   different values of $a$ varies relying on the dimension of predictors and the correlation structure of the predictors,   fixing   $a$  is difficult. Thus  either giving  a hyperprior for $a$ or using the $R^2$ method proposed in Section  \ref{section-tuning hyperparameters} to tune $a$ is suggested.   


\begin{table} [h!]
	\caption{Average squared error for the posterior mean,       given  Dirichlet-Laplace prior with   $a$ fixed at $\frac{1}{2}$, $\frac{1}{n}$ and $\frac{1}{p}$, based on 200 datasets with standard errors in parentheses. }
	\label{table-dl2np}
	\renewcommand{\arraystretch}{1.0}
	\begin{center}\small\addtolength{\tabcolsep}{-3pt}
		\begin{tabular}{ c|    ccc  | ccc | ccc  }
			\hline\hline
			&   \multicolumn{3}{c|}{$p=50$}   &  \multicolumn{3}{c|}{$p=500$} &  \multicolumn{3}{c}{$p=1000$}  \\ 
			$a$  & $ \frac{1}{2}$ &  $\frac{1}{n}$  &  $\frac{1}{p}$ 
			& $\frac{1}{2}$ &  $\frac{1}{n}$  &  $\frac{1}{p}$ 
			& $\frac{1}{2}$ &  $\frac{1}{n}$  &  $\frac{1}{p}$ 
			\\
			\hline
			$\rho=0.5$    & 0.772  & 0.877   & 0.874 & 1.292   & 1.400   & 1.953 &  1.470    & 1.434  & 2.196 \\
			&  {\footnotesize $( 0.0234 )$ } & {\footnotesize $( 0.0325 )$ } & {\footnotesize $( 0.0329 )$ } &
			{\footnotesize $( 0.0421 )$ } & {\footnotesize $( 0.0519 )$ } & {\footnotesize $( 0.0576 )$ }  &
			{\footnotesize $( 0.0451 )$ } & {\footnotesize $( 0.1070 )$ } & {\footnotesize $( 0.1196 )$ }  \\ \hline
			$\rho=0.9$    & 1.989  & 1.751   & 1.715 & 2.193   & 2.142   & 2.546 &  2.299    & 2.247  & 2.426  \\
			&{\footnotesize $( 0.0559 )$ } & {\footnotesize $( 0.0737 )$ } & {\footnotesize $( 0.0739 )$ } &
			{\footnotesize $( 0.0767 )$ } & {\footnotesize $( 0.0981 )$ } & {\footnotesize $( 0.1180 )$ } & 
			{\footnotesize $( 0.1101 )$ } & {\footnotesize $( 0.1178 )$ } & {\footnotesize $( 0.1186 )$ } \\
			\hline 
		\end{tabular}
	\end{center}
\end{table}

Furthermore,   to verify Theorem \ref{theorem_normal tuning} described  in Section \ref{section-tuning hyperparameters},  additional calculations were performed. 
For each of the above 200 datasets,  `Normal\_tune''  returns a  ``best''  tuned  $\gamma$ through conducting the practical procedures as introduced in Section \ref{section-tuning hyperparameters}, 
and we name it  as ``Tuned''.  
Also, by Theorem \ref{theorem_normal tuning},  the theoretic ``best'' $\gamma$ can be  derived based on the eigenvalues of $\bm{X}^T\bm{X}/n$ for each dataset, and we name it  as ``Derived''. 
In addition,  for each of the above 200 datasets, the design matrix $X$ is generated from a multivariate normal distribution with specific and fixed covariance structure. So the eigenvalues of such true covariance matrix,  instead of $\bm{X}^T\bm{X}/n$, 
can be used to derive the theoretic ``best'' $\gamma$, and we name it as ``Theoretic'' value. 
Table \ref{table-tau-tune-theoretic}  gives the ``Theoretic'' value,  and the mean  of  ``Derived'' and ``Tuned'' value together with the standard error    among the 200 datasets, for simulation setups $\rho=0.5$ and $0.9$.  In general, the three values  are similar and all of them  are close to the value of $p$.  So in practice,  $\gamma$ can be set as  the ``Derived'' value  based on the eigenvalues of $\bm{X}^T\bm{X}/n$, or for simplicity, $\gamma=p$ can also be used.     


\begin{table} [h!]
	\caption { Theoretic   $\gamma$  in the normal prior (\ref{normal-prior}) based on Theorem \ref{theorem_normal tuning},  together with   mean    of the derived and tuned  $\gamma$ through methods proposed in Section \ref{section-tuning hyperparameters}, based on 200 datasets with standard errors in parentheses. } 
	\label{table-tau-tune-theoretic}
	\renewcommand{\arraystretch}{1.0}
	\begin{center}  \small\addtolength{\tabcolsep}{-1.5pt}
		\begin{tabular}{ l|    c cc  |c  cc    }
			\hline \hline 
			& \multicolumn{3}{c| } { $\rho=0.5$} & \multicolumn{3}{c } { $\rho=0.9$}\\
			&\multicolumn{1}{c } {Theoretic} & Derived    &\multicolumn{1}{c| } {Tuned }
			&\multicolumn{1}{c } {Theoretic} & Derived    &\multicolumn{1}{c } {Tuned }
			\\
			\hline
			$p=50$ &   47.6  & 46.6 (0.11) & 48.6 (0.24) & 40.8 & 39.9 (0.18) & 41.3 (0.25) \\
			$p=500$ &   490.0  & 481.8 (0.35)& 474.3 (0.83) & 482.4 & 474.1 (0.81) & 471.9 (1.23) \\
			$p=1000$ &   981.7 & 965.1 (0.51) &  947.1 (1.50) & 974.0  & 956.9 (1.18)& 944.3 (1.82)\\
			\hline 
		\end{tabular} 
	\end{center}
\end{table}
\section{Real Data Analysis} \label{sec_real data}
We now analyze data on mouse gene expression from the experiment conducted  by  \cite{lan2006combined}. 
There were 60 arrays to monitor the expression levels of $22,575$ genes consisting of  $31$ female and $29$ male mice. 
Quantitative real-time PCR were used to measure 
some physiological phenotypes,
including numbers of phosphoenopyruvate carboxykinase (PEPCK), glycerol-3-phosphate
acyltransferase (GPAT), and stearoyl-CoA desaturase 1 (SCD1). 
The gene expression data and the phenotypic data can be found at GEO 
(http://www.ncbi.nlm.nih.gov/geo; accession number GSE3330).

First,  by ordering the magnitude of marginal correlation between the genes with the three responses from the largest to the smallest, 
$22,575$ genes were screened down to the $999$ genes,  thus reducing the number of candidate predictors of the three linear regressions. Note that 
the top $999$ genes were not the same for  the $3$ responses. 
Then  for each of the $3$
regressions, the dataset is  composed  of  $n =60$ observations and $p =1, 000$ predictors (gender along with the $999$ genes). 
After the  screening, the Lasso estimator and the penalized credible region method applied on the normal, Laplace and DL priors were used.  The hyperparameters in those prior distributions are tuned   through the  $R^2$ method introduced in Section  \ref{section-tuning hyperparameters}, with a target of uniform distribution.

To evaluate the performance of the proposed approach,  the first step was to randomly split the sample size  $60$  into a training set of size $55$ and a testing set of size $5$. 
The stopping rule was  BIC. To be more specific, the selected model was the one with smallest BIC among all models in which the  number of predictors  is less than $30$. 
Then   the selected model was used to predict the remaining $5$ observations, and the prediction error was then obtained. We repeated this for   $100$ replicates   in order to compare the prediction errors. 
Table \ref{table-real data} shows the mean squared prediction error (with its standard error) based on the $100$ random splits of the data.  The mean selected model size    (with its standard error) is also included.

\begin{table} [h!]
	\caption {Mean squared prediction error and model size, with standard errors in parenthesis, based on 100 random splits of the real data.    } 
	\label{table-real data}
	\renewcommand{\arraystretch}{1.0}
	\begin{center}
		\small\addtolength{\tabcolsep}{-1pt}
		\begin{tabular}{ l|    cc | cc  | cc    }
			\hline \hline 
			&\multicolumn{2}{c| } { PEPCK}& \multicolumn{2}{c| } {GPAT} & \multicolumn{2}{c } { SCD1}\\
			&\multicolumn{1}{c } {MSPE } &\multicolumn{1}{c| } {Model Size}  
			&\multicolumn{1}{c } {MSPE } &\multicolumn{1}{c| } {Model Size}  
			&\multicolumn{1}{c } {MSPE } &\multicolumn{1}{c } {Model Size}  
			\\
			\hline
			Lasso & 0.54 (0.026) & 25.8 (0.34) & 1.43 (0.082)  & 24.4 (0.56) &0.55 (0.052) &  26.1 (0.33)\\
			Normal  & 0.66 (0.033) & 16.8 (0.67) & 1.30 (0.099) &  16.3 (0.66) & 0.71 (0.059)& 10.8 (0.60) \\
			Laplace & 0.70 (0.037) & 17.0 (0.78) & 1.19 (0.086) &  21.4 (0.56) & 0.69 (0.054) &  14.8 (0.82)\\
			DL  & 0.49 (0.032) & 18.4 (0.73) & 1.37 (0.102) & 13.1 (0.68) & 0.54 (0.037) &  14.0 (0.59) \\
			\hline 
		\end{tabular}
	\end{center}
\end{table}

Overall, the results show that the proposed   penalized credible region selection method using global-local shrinkage priors such as  DL prior performs  well.   For all $3$ responses,  the penalized credible region approach with DL prior performs better than the Lasso estimator and  has a smaller number of predictors. For PEPCK and SCD1, the DL prior has significant better performance than the normal prior and Laplace prior.  For GPAT, there is no significant difference between normal and DL prior.  In all, for this dataset, the proposed approach generally improves the performance by replacing the normal prior with the DL prior.

\section{Discussion} \label{sec_discussion}
In this paper, we  extend the penalized credible variable selection approach  by using   global-local shrinkage priors. 
Simulation studies  show that  the GL shrinkage priors  outperform the original normal prior. 
Our main result also includes modifying the Dirichlet-Laplace prior to accommodate  the linear regression model instead of the simple normal mean problem as in \cite{bhattacharya2015dirichlet}.   In  theory, we  obtain the selection consistency  for the penalized  credible region method using the global-local shrinkage priors  when $p=o(n)$. 
Posterior consistency for the normal and DL priors are also shown.

Furthermore, this paper   introduces a new  default method to tune the hyperparameters in   prior distributions based on the induced prior distribution of $R^2$.  The hyperparameter is chosen to minimize a discrepancy between the induced distribution of $R^2$ and a default Beta distribution. 
For the  normal prior, a closed form of the    hyperparameters  is derived. This method is straightforward  and   efficient as it only involves  the   prior distributions. 
A simulation study  illustrates that our proposed tuning method improves upon the usual hyperprior method.

\bibliographystyle{agsm}
\bibliography{reference}

\newpage 
\section*{Appendix: Proofs} \label{sec_proof_credible}
\textbf{Proof of Theorem \ref{theorem_normal_consistency}}
\begin{proof}
	According to  Theorem 1 in \cite{armagan2013posterior}, if  under a particular prior, ${\bm\beta}_n$ satisfies 
	\[
	P({\bm\beta}_n:||{\bm\beta}_n-{\bm\beta}_n^0|| < \frac{\Delta}{n^{\rho/2}}) >\exp  (-dn)
	\]
	for all $0<\Delta<\varepsilon^2 d_{\min} /(48d_{\max})$ and $0<d<\varepsilon^2 d_{\min} /(32\sigma^2) - 3\Delta d_{\max}/(2\sigma^2)$ and some $\rho>0$, then 
	the posterior of ${\bm\beta}_n$   is  consistent. 
	So to get the posterior consistency, the key  is to calculate the probability of $\{{\bm\beta}_n:||{\bm\beta}_n-{\bm\beta}_n^0|| < \frac{\Delta}{n^{\rho/2}} \}$ under the given prior. 
	
	Following the proof of Theorem 2 in \cite{armagan2013posterior}, we have 
	{	\begin{eqnarray}\label{eq_6}
			&& P({\bm\beta}_n:||{\bm\beta}_n-{\bm\beta}_n^0|| < \frac{\Delta}{n^{\rho/2}}) = 
			P\left\{{\bm\beta}_n: \sum_{j\in \mathcal{A}^0_n } (\beta_{nj}-\beta_{nj}^0)^2 + \sum_{j\not\in \mathcal{A}^0_n} \beta_{nj}^2 < \frac{\Delta^2}{n^{\rho}}\right\} \nonumber \\
			&\geq& 
			\prod\limits_{j\in\mathcal{A}^0_n} \left\{ P\left( \beta_{nj}:   |\beta_{nj}-\beta_{nj}^0| <\frac{\Delta}{\sqrt{p_n}n^{\rho/2}}\right)\right\} \times 
			P\left\{{\bm\beta}_{n} : \sum_{j\not\in \mathcal{A}^0_n}
			\beta_{nj}^2 < \frac{(p_n-q_n)\Delta^2}{p_nn^{\rho}} \right\}   \nonumber \\
			%
			&\geq & \prod\limits_{j\in\mathcal{A}^0_n} \left\{ P\left( \beta_{nj}^0 -\frac{\Delta}{\sqrt{p_n}n^{\rho/2}} <   \beta_{nj} < \beta_{nj}^0 +\frac{\Delta}{\sqrt{p_n}n^{\rho/2}}\right)\right\} \times 
			\left\{1-  \frac{p_n n^\rho E( \sum_{j\notin \mathcal{A}^0_n}\beta_{nj}^2)}{(p_n-q_n)\Delta^2} \right\} \nonumber \\
			&\geq & \prod\limits_{j\in\mathcal{A}^0_n} \left\{ P\left(-\sup_{j\in\mathcal{A}^0_n}|\beta_{nj}^0| -\frac{\Delta}{\sqrt{p_n}n^{\rho/2}} <   \beta_{nj} < \sup_{j\in\mathcal{A}^0_n}|\beta_{nj}^0| +\frac{\Delta}{\sqrt{p_n}n^{\rho/2}}\right)\right\} \times  \nonumber \\
			&& 	\left\{1-  \frac{p_n n^\rho E( \sum_{j\notin \mathcal{A}^0_n}\beta_{nj}^2)}{(p_n-q_n)\Delta^2} \right\} \nonumber \\
			&\geq &  \left\{  2\frac{\Delta}{\sqrt{p_n}n^{\rho/2}} f(\sup_{j\in\mathcal{A}^0_n}|\beta_{nj}^0| + \frac{\Delta}{\sqrt{p_n}n^{\rho/2}} ) \right\}^{q_n} \times 
			\left\{1-  \frac{p_n n^\rho E( \sum_{j\notin \mathcal{A}^0_n}\beta_{nj}^2)}{(p_n-q_n)\Delta^2} \right\} , 
		\end{eqnarray}}
		where $f$ is the prior pdf of $\bm\beta$,   symmetric and decreasing when the support is positive. 
		In  normal prior (\ref{normal-prior}), $f(\beta_{nj}) =\frac{1}{\sqrt{2\pi \sigma^2/\gamma_n}} \exp  \{-\frac{\beta_{nj}^2}{2\sigma^2/\gamma_n}\} $, $E(\beta_{nj}^2) = \sigma^2/\gamma_n$. Following from (\ref{eq_6}), we have
		\begin{eqnarray} \label{9}
			&& P({\bm\beta}_n:||{\bm\beta}_n-{\bm\beta}_n^0|| < \frac{\Delta}{n^{\rho/2}}) \nonumber\\
			&\geq & \left\{  2\frac{\Delta}{\sqrt{p_n}n^{\rho/2}}
			\frac{1}{\sqrt{2\pi \sigma^2/\gamma_n}} \exp  \{-\frac{(\sup_{j\in\mathcal{A}^0_n}|\beta_{nj}^0| + \frac{\Delta}{\sqrt{p_n}n^{\rho/2}})^2}{2\sigma^2/\gamma_n}\}
			\right\}^{q_n} \times 
			\left\{1-  \frac{p_n n^\rho \sigma^2/\gamma_n }{ \Delta^2}\right\}  \nonumber\\
			&\geq& \left\{   \frac{2 \Delta}{\sqrt{p_n}n^{\rho/2}}
			\frac{1}{\sqrt{2\pi \sigma^2/\gamma_n}} \exp  \{-\frac{ (\sup_{j\in\mathcal{A}^0_n}|\beta_{nj}^0|)^2 +  (\frac{\Delta}{\sqrt{p_n}n^{\rho/2}})^2}{\sigma^2/\gamma_n}\}
			\right\}^{q_n}  
			\left\{1-  \frac{p_n n^\rho \sigma^2/\gamma_n }{ \Delta^2}\right\}. \nonumber 
		\end{eqnarray}
		Taking the negative logarithm of both sides of the above formula, and letting $\sqrt{\sigma^2/\gamma_n} = C/(\sqrt{p_n}n^{\rho/2}\log n)$,  we have 
		{  	\begin{eqnarray} \label{8}
				&& -\log P({\bm\beta}_n:||{\bm\beta}_n-{\bm\beta}_n^0|| < \frac{\Delta}{n^{\rho/2}}) \nonumber\\
				&\leq& -q_n\log \left\{  \frac{2\Delta}{\sqrt{p_n}n^{\rho/2}}
				\frac{1}{\sqrt{2\pi \sigma^2/\gamma_n}}   \right\}
				+q_n \frac{(\sup_{j\in\mathcal{A}^0_n}|\beta_{nj}^0|)^2 + (\frac{\Delta}{\sqrt{p_n}n^{\rho/2}})^2}{ \sigma^2/\gamma_n}  \nonumber \\
				&&	-
				\log \left\{1-  \frac{p_n n^\rho \sigma^2/\gamma_n }{ \Delta^2} \right\}    \nonumber  \\
				&=&-q_n\log \left\{  \frac{2\Delta\log n}{\sqrt{2\pi}C}\right\}
				+q_n \frac{(\sup_{j\in\mathcal{A}^0_n}|\beta_{nj}^0|)^2 + (\frac{\Delta^2}{p_nn^{\rho}})}{ C^2/(p_n n^\rho (\log n)^2)}  - \log \left\{1-  \frac{C^2 }{ \Delta^2(\log n)^2} \right\}    \nonumber  \\      
				&=& -q_n\log \left\{  \frac{2\Delta\log n}{\sqrt{2\pi}C}\right\}
				- \log \left\{1-  \frac{C^2 }{ \Delta^2(\log n)^2} \right\}  
				+ \frac{q_n\Delta^2(\log n)^2}{C^2}  \nonumber \\
				&& +\frac{q_n  p_nn^\rho (\log n)^2 (\sup_{j\in\mathcal{A}^0_n}|\beta_{nj}^0|)^2}{C^2} .    \nonumber 
			\end{eqnarray}}
			The last term is the  dominating one in  the above equation, and 
			$ -\log P({\bm\beta}_n:||{\bm\beta}_n-{\bm\beta}_n^0|| < \frac{\Delta}{n^{\rho/2}}) <dn$ for all $d>0$ if $q_n  p_nn^\rho (\log n)^2 = o(n) $, i.e.,  $q_n=o(\frac{n^{1-\rho}} {p_n(\log n)^2} )$. So given the normal prior and assumptions    \ref{a_error}-\ref{a_beta0},  if   $q_n=o(\frac{n^{1-\rho}} {p_n(\log n)^2} ) $ for $\rho \in (0,1)$, the prior  satisfies  $P({\bm\beta}_n:||{\bm\beta}_n-{\bm\beta}_n^0||  < \frac{\Delta}{n^{\rho/2}})  > \exp(-dn)$. The posterior consistency is completed by Theorem 1 in \cite{armagan2013posterior}. 
		\end{proof}

		\textbf{Proof of Theorem \ref{theorem_DL_consistency}}
		\begin{proof}
			According to Section  3 in  \cite{bhattacharya2015dirichlet}, 
			the Dirichlet-Laplace prior (\ref{eq_DL-DE-form})
			can also  be represented as 
			\[ 
			\begin{array}{l l}
			\beta_j| \xi_j \sim \text{DE}(\xi_j\sigma),  \   
			\xi_j\sim  \text{Ga}(a,1/2).  
			\end{array} 
			\]
			And the marginal distribution of $\beta_j$ is 
			\begin{eqnarray*}
				f_d(\beta_j) & = & \int_{\xi_j=0}^{\infty}
				\left[  \frac{1}{2\xi_j \sigma}\exp\{-\frac{|\beta_j|} {\xi_j\sigma }\}\right]    
				\left[ \frac{(\frac{1}{2})^a}{\Gamma(a)}  \xi_j^{a-1}\exp\{-\frac{1}{2}\xi_j\}\right] \,  d\xi_j \\
				& = & \frac{(\frac{1}{2})^a}{2\Gamma(a)\sigma} \int_{\xi_j=0}^{\infty} \exp\{-\frac{|\beta_j|} {\xi_j\sigma} \} \xi_j^{a-2}\exp\{-\frac{1}{2}\xi_j\} \, d\xi_j. 
			\end{eqnarray*}
			Without loss of generality, we   assume $\sigma=1$. According to the Proposition 3.1 in \cite{bhattacharya2015dirichlet}, the marginal density function   of $\beta_j$ for any $1\leq j\leq p$ is 
			\begin{eqnarray*}
				f_d(\beta_j) = \frac{1}{2^{(1+a)/2}\Gamma(a)} |\beta_j|^{(a-1)/2} K_{1-a}(\sqrt{2|\beta_j|}), 
			\end{eqnarray*}
			where 
			\begin{eqnarray*}
				K_\nu(x) = \frac{\Gamma(\nu+1/2) (2x)^\nu}{\sqrt{\pi}} \int_{0}^{\infty} \frac{\cos t}{(t^2+x^2)^{\nu+1/2}}\, dt
			\end{eqnarray*}
			is the modified Bessel function of the second kind.  
			
			Also we have 
			\begin{eqnarray*}
				E(\beta_j^2) &=& \int_{\beta_j=-\infty}^{\infty} \int_{\xi_j=0}^{\infty}
				\left[  \frac{1}{2\xi_j } \beta_j^2 \exp\{-\frac{|\beta_j|} {\xi_j}\}\right]    
				\left[ \frac{(\frac{1}{2})^a}{\Gamma(a)}  \xi_j^{a-1}\exp\{-\frac{1}{2}\xi_j\}\right] \, d\xi_j d\beta_j \\
				&=&  \int_{\xi_j=0}^{\infty}  2\xi_j^2  \left[ \frac{(\frac{1}{2})^a}{\Gamma(a)}  \xi_j^{a-1}\exp\{-\frac{1}{2}\xi_j\}\right] \,  d\xi_j   
				= 8a(a+1). 
			\end{eqnarray*}

			
			Facts:   (i)  10.37.1 in \cite{NISTDLMF}, if $0\leq \nu <\mu$ and $z$ is a real number, then $|K_\nu(z)|<|K_\mu(z)|$; (ii)  10.39.2  in \cite{NISTDLMF}, when $\nu=\frac{1}{2}$, $K_{\frac{1}{2}}(z) =\sqrt{\frac{\pi}{2z}}e^{-z}$. 
			For $K_{1-a}(\sqrt{2|\beta_j|})$, as $a\leq \frac{1}{2} $,  or  $1-a\geq \frac{1}{2}$; also $\sqrt{2|\beta_j|} $ is a real number for $j=1,\cdots, p$. So when $\beta_j$ is fixed,  $K_{1-a}(\sqrt{2|\beta_j|}) \geq K_{\frac{1}{2}}(\sqrt{2|\beta_j|} ) = \sqrt{\frac{\pi}{2\sqrt{2|\beta_j|}}}\exp\{-\sqrt{2|\beta_j|}\}$. 
			Combining with the fact that   $\Gamma(a)=  a^{-1} -\gamma_0+ O(a)\leq a^{-1}$ for $a$ close to zero,  where $\gamma_0$ is the Euler-Mascheroni constant, then  for $j=1,\cdots, p$, we have 
			{  	\begin{eqnarray*}
					f_d(\beta_j)  &\geq & \frac{1}{2^{(1+a)/2}\Gamma(a)} |\beta_j|^{(a-1)/2}  \sqrt{\frac{\pi}{2\times  \sqrt{2|\beta_j|}}}\exp\{-\sqrt{2|\beta_j|}\}  \\
					&\geq &\frac{\sqrt{\pi}a |\beta_j|^{(2a-3)/4}  \exp\{-\sqrt{2|\beta_j|}\}}{2^{(2a+5)/4} } . 
				\end{eqnarray*}}
				Then following from (\ref{eq_6}), 
				{ 
					\begin{eqnarray} \label{Diri-main}
						&& P(\bm\beta_n:||{\bm\beta}_n-{\bm\beta}_n^0|| < \frac{\Delta}{n^{\rho/2}}) \nonumber\\
						&\geq &  \left\{  2\frac{\Delta}{\sqrt{p_n}n^{\rho/2}} f_d(\sup_{j\in\mathcal{A}^0_n}|\beta_{nj}^0| + \frac{\Delta}{\sqrt{p_n}n^{\rho/2}} ) \right\}^{q_n} \times 
						\left\{1-  \frac{p_n n^\rho E( \sum_{j\notin \mathcal{A}^0_n}\beta_{nj}^2)}{(p_n-q_n)\Delta^2} \right\}   \nonumber \\
						%
						&\geq &  \left\{ \frac{2\Delta}{\sqrt{p_n}n^{\rho/2}} 
						\frac{\sqrt{\pi}a_n |\sup_{j\in\mathcal{A}^0_n}|\beta_{nj}^0| + \frac{\Delta}{\sqrt{p_n}n^{\rho/2}}|^{(2a-3)/4}  e^{-\sqrt{2|\sup_{j\in\mathcal{A}^0_n}|\beta_{nj}^0| + \frac{\Delta}{\sqrt{p_n}n^{\rho/2}}|}}}{2^{(2a+5)/4}}
						\right\}^{q_n}  \nonumber \\
						&&
						\left\{1-  \frac{p_n n^\rho 8a_n(a_n+1)}{\Delta^2} \right\}     \nonumber \\
						&\geq &  \left\{ \frac{\sqrt{\pi}  \Delta a_n
							(\sup_{j\in\mathcal{A}^0_n}|\beta_{nj}^0| + \frac{\Delta}{\sqrt{p_n}n^{\rho/2}})^{(2a-3)/4} }{ 2^{(2a+1)/4}\sqrt{p_n}n^{\rho/2}} 
						\exp\{-\sqrt{2(\sup_{j\in\mathcal{A}^0_n}|\beta_{nj}^0| + \frac{\Delta}{\sqrt{p_n}n^{\rho/2}})}\} 
						\right\}^{q_n}   \nonumber \\
						&&	\left\{1-  \frac{p_n n^\rho 16a }{\Delta^2} \right\}    \nonumber. 
					\end{eqnarray}}
					Taking the negative logarithm of both sides of  the above formula,  and letting $a_n = C/( p_n n^{\rho}\log n)$, we have  
					{  \begin{eqnarray}
							&& -\log P({\bm\beta}_n:||{\bm\beta}_n-{\bm\beta}_n^0|| < \frac{\Delta}{n^{\rho/2}}) \nonumber\\
							&\leq & -q_n \log  \frac{\sqrt{\pi}  \Delta a_n
								(\sup_{j\in\mathcal{A}^0_n}|\beta_{nj}^0| + \frac{\Delta}{\sqrt{p_n}n^{\rho/2}})^{(2a_n-3)/4} }{ 2^{(2a_n+1)/4}\sqrt{p_n}n^{\rho/2}} + q_n\sqrt{2(\sup_{j\in\mathcal{A}^0_n}|\beta_{nj}^0| + \frac{\Delta}{\sqrt{p_n}n^{\rho/2}}) } \nonumber \\
							&&	-  \log \left\{1-  \frac{p_n n^\rho 16a_n }{\Delta^2} \right\}     \nonumber \\
							%
							& =  &  -q_n\log \frac{\sqrt{\pi}  \Delta a_n}{\sqrt{p_n}n^{\rho/2}} 
							-\frac{q_n a_n}{2}  \log  (\sup_{j\in\mathcal{A}^0_n}|\beta_{nj}^0| + \frac{\Delta}{\sqrt{p_n}n^{\rho/2}}) 
							-  \log \left\{1-  \frac{p_n n^\rho 16a_n }{\Delta^2} \right\}   \nonumber \\
							&&        + q_n\sqrt{2(\sup_{j\in\mathcal{A}^0_n}|\beta_{nj}^0| + \frac{\Delta}{\sqrt{p_n}n^{\rho/2}}) }
							+\frac{3q_n}{4} \log  (\sup_{j\in\mathcal{A}^0_n}|\beta_{nj}^0| + \frac{\Delta}{\sqrt{p_n}n^{\rho/2}})  + \frac{q_n (2a_n+1) \log 2}{4}  \nonumber\\ 
							\nonumber \\
							& = &  -q_n\log  (\sqrt{\pi}  \Delta C )
							-\frac{q_n C}{2 p_n n^{\rho}\log n}  \log  (\sup_{j\in\mathcal{A}^0_n}|\beta_{nj}^0| + \frac{\Delta}{\sqrt{p_n}n^{\rho/2}}) 
							-  \log \left\{1-  \frac{  16C }{\Delta^2 \log n} \right\}   
							\nonumber \\
							&&	+ q_n\sqrt{2(\sup_{j\in\mathcal{A}^0_n}|\beta_{nj}^0|  + \frac{\Delta}{\sqrt{p_n}n^{\rho/2}}) }
							+\frac{3q_n}{4} \log  (\sup_{j\in\mathcal{A}^0_n}|\beta_{nj}^0| + \frac{\Delta}{\sqrt{p_n}n^{\rho/2}})   \nonumber \\
							&&	+ \frac{q_n  \log 2}{4} \left(\frac{2C}{p_n n^{\rho}\log n}+1\right) 
							+ q_n\log (p_n^{3/2}n^{3\rho/2} \log n ) .    \nonumber
						\end{eqnarray}} 
						The last term is the  dominating one in  the above equation,  and $  -\log P({\bm\beta}_n:||{\bm\beta}_n-{\bm\beta}_n^0|| < \frac{\Delta}{n^{\rho/2}}) < dn$ for all $d>0$ if $q_n = o(  n /\log (p_n^{3/2} n^{3\rho/2} \log n))$. Furthermore,   $q_n = o(n/\log n)$ is a sufficient condition. 
						So 
						given the DL prior and assumptions \ref{a_error}-\ref{a_beta0}, if   $q_n = o( n/ \log n)$, the prior  satisfies  $P({\bm\beta}_n:||{\bm\beta}_n-{\bm\beta}_n^0||  < \frac{\Delta}{n^{\rho/2}})  > \exp(-dn)$. The posterior consistency is completed by Theorem 1 in \cite{armagan2013posterior}. 
					\end{proof}

					\begin{lemma} \label{lemma1}
						Given assumptions  \ref{a_error}-\ref{a_min}, if     $c_n/p_n\rightarrow\infty$ and $\gamma_n=o(n)$, then the true parameter ${\bm\beta}_n^0$  is contained in the proposed region, i.e., $({\bm\beta}_n-\hat{{\bm\beta}}_n)^T\bm\Sigma^{-1}_n({\bm\beta}_n-\hat{{\bm\beta}}_n)\leq c_n$,  with probability increasing to 1.
					\end{lemma}
					\begin{proof}
						Denote $\bm{\varepsilon}_n= (\varepsilon_1,\cdots,\varepsilon_n)^T$.  
						Since $\hat{{\bm\beta}}_n =(\bm{X}_n^T\bm{X}_n+\gamma_n\bm{I}_n)^{-1}(\bm{X}_n^T\bm{Y}_n) $, and $\bm\Sigma^{-1}_n = (\hat{\sigma}_n^2)^{-1}(\bm{X}_n^T\bm{X}_n+\gamma_n\bm{I}_n)$ with $\hat{\sigma}_n^2\rightarrow\sigma^2$, 
						then
						$\hat{{\bm\beta}}_n-{\bm\beta}_n^0 = (\bm{X}_n^T\bm{X}_n+\gamma_n\bm{I}_n)^{-1}(\bm{X}_n^T(\bm{X}_n{\bm\beta}_n^0+\bm{\varepsilon}_n)) - {\bm\beta}_n^0 =  \left[(\frac{\bm{X}_n^T\bm{X}_n}{n}+\frac{\gamma_n\bm{I}_n}{n})^{-1}\frac{\bm{X}_n^T\bm{X}_n}{n}  - \bm{I}_n\right]{\bm\beta}_n^0 + (\frac{\bm{X}_n^T\bm{X}_n}{n}+\frac{\gamma_n\bm{I}_n}{n})^{-1}\frac{\bm X^T}{n}\bm{\varepsilon}_n
						$.
						Note $(\frac{\bm{X}_n^T\bm{X}_n}{n}+\frac{\gamma_n\bm{I}_n}{n})^{-1}\frac{\bm{X}_n^T\bm{X}_n}{n}  - \bm{I}_n = -\frac{\gamma_n}{n}(\frac{\bm{X}_n^T\bm{X}_n}{n}+\frac{\gamma_n\bm{I}_n}{n})^{-1}$. So for each fixed $n$, $\hat{\bm\beta}_n-{\bm\beta}^0_n\sim N({\bm m}_n,\bm V_n)$, where \[
						{\bm{m}_n} = -\frac{\gamma_n}{n}(\frac{\bm{X}_n^T\bm{X}_n}{n}+\frac{\gamma_n\bm{I}_n}{n})^{-1}{\bm\beta}^0_n
						\] and 
						\[
						\bm V_n= \frac{\sigma^2}{n}(\frac{\bm{X}_n^T\bm{X}_n}{n}+\frac{\gamma_n\bm{I}_n}{n})^{-1}\frac{\bm{X}_n^T\bm{X}_n}{n}(\frac{\bm{X}_n^T\bm{X}_n}{n}+\frac{\gamma_n\bm{I}_n}{n})^{-1}. 
						\]
						Then $(\hat{\bm\beta}_n-{\bm\beta}^0_n-{\bm m}_n)^T \bm V^{-1}_n(\hat{\bm\beta}_n-{\bm\beta}_n^0-{\bm m}_n)\sim \chi^2_{p_n}$ for each fixed $n$.  Further,  
						\begin{equation} \label{chi2/p}
							\lim\limits_{p_n, n\rightarrow\infty} \frac{1}{p_n} (\hat{\bm\beta}_n-{\bm\beta}_n^0-{\bm m}_n)^T\bm V_n^{-1}(\hat{\bm\beta}_n-{\bm\beta}_n^0-{\bm m}_n) =  \lim\limits_{p_n, n\rightarrow\infty} \frac{\chi^2_{p_n}} {p_n} = 1.
						\end{equation}
						%
						Furthermore, 
						{\small 	\begin{eqnarray} \label{formulat}
								&& (\hat{\bm\beta}_n-{\bm\beta}_n^0-{\bm m}_n)^T\bm V_n^{-1}(\hat{\bm\beta}_n-{\bm\beta}_n^0-{\bm m}_n)  \nonumber \\
								&=& (\hat{\bm\beta}_n-{\bm\beta}_n^0)^T\bm V_n^{-1}(\hat{\bm\beta}_n-{\bm\beta}_n^0)
								-2{\bm m}_n^T\bm V_n^{-1}(\hat{\bm\beta}_n-{\bm\beta}_n^0) + {\bm m}_n^T\bm V_n^{-1}{\bm m}_n \nonumber \\
								&=& (\hat{\bm\beta}_n-{\bm\beta}_n^0)^T\bm V_n^{-1}(\hat{\bm\beta}_n-{\bm\beta}_n^0)
								-2{\bm m}_n^T\bm V_n^{-1}\left( {\bm m}_n + (\frac{\bm{X}_n^T\bm{X}_n}{n}+\frac{\gamma_n\bm{I}_n}{n})^{-1}\frac{\bm{X}_n^T}{n}\bm{\varepsilon}_n \right) + {\bm m}_n^T\bm V_n^{-1}{\bm m}_n \nonumber  \\
								&=& (\hat{\bm\beta}_n-{\bm\beta}_n^0)^T\bm V_n^{-1}(\hat{\bm\beta}_n-{\bm\beta}_n^0)
								-{\bm m}_n^T\bm V_n^{-1}{\bm m}_n
								- 2{\bm m}_n^T\bm V_n^{-1}(\frac{\bm{X}_n^T\bm{X}_n}{n}+\frac{\gamma_n\bm{I}_n}{n})^{-1}\frac{\bm{X}_n^T}{n}\bm{\varepsilon}_n . 
							\end{eqnarray}}
							
							First of all,  since 
							{ 	\begin{eqnarray*}
									&& 0 \leq \lim\limits_{p_n, n\rightarrow\infty} \frac{1}{p_n}{\bm m}_n^T\bm V_n^{-1}{\bm m}_n  \\
									&=& \lim\limits_{p_n, n\rightarrow\infty}  \frac{1}{p_n}
									\frac{n}{\sigma^2}  \frac{\gamma_n^2}{n^2}
									{\bm\beta}^{0T}_n   (\frac{\bm{X}_n^T\bm{X}_n}{n}+\frac{\gamma_n\bm{I}_n}{n})^{-1}   (\frac{\bm{X}_n^T\bm{X}_n}{n}+\frac{\gamma_n\bm{I}_n}{n})
									\\ && (\frac{\bm{X}_n^T\bm{X}_n}{n})^{-1}(\frac{\bm{X}_n^T\bm{X}_n}{n}+\frac{\gamma_n\bm{I}_n}{n}) (\frac{\bm{X}_n^T\bm{X}_n}{n}+\frac{\gamma_n\bm{I}_n}{n})^{-1} {\bm\beta}^0_n \\
									&=& \lim\limits_{p_n, n\rightarrow\infty} \frac{1}{p_n} \frac{\gamma_n^2}{n\sigma^2}{\bm\beta}^{0T}_n    (\frac{\bm{X}_n^T\bm{X}_n}{n})^{-1}{\bm\beta}^0_n 
									=  \lim\limits_{p_n, n\rightarrow\infty} \frac{1}{p_n} \frac{\gamma_n^2}{n\sigma^2}{\bm\beta}^{0T}_n \bm\Gamma_n \bm D_n^{-1}\bm\Gamma_n^{T}  {\bm\beta}^0_n \\
									&\leq & \lim\limits_{p_n, n\rightarrow\infty} \frac{1}{p_n} \frac{\gamma_n^2}{n\sigma^2}{\bm\beta}_n^{0T} \bm\Gamma_n \text{diag}\{1/d_{\min}, \cdots, 1/d_{\min} \} \bm\Gamma_n^{T}  {\bm\beta}_n^0   \\
									&=& \lim\limits_{p_n, n\rightarrow\infty} \frac{1}{p_n} \frac{\gamma_n^2}{n\sigma^2} \frac{1}{d_{\min}} ||{\bm\beta}^{0}_n||^2 
									= 0  , 
								\end{eqnarray*}}
								so  we have
								\begin{eqnarray} \label{equation-te1}
									\lim\limits_{p_n, n\rightarrow\infty} \frac{1}{p_n}{\bm m}_n^T\bm V_n^{-1}{\bm m}_n = 0. 
								\end{eqnarray} 
								
								Next,  let's get the limit of 
								$
								- \frac{1}{p_n} {\bm m}_n^T\bm V_n^{-1}(\frac{\bm{X}_n^T\bm{X}_n}{n}+\frac{\gamma_n\bm{I}_n}{n})^{-1}\frac{\bm  X^T}{n}\bm{\varepsilon}_n  $, or equivalently, the limit of $ 
								\frac{1}{p_n}   \frac{\gamma_n}{\sigma^2}
								{\bm\beta}_n^{0T}(\frac{\bm{X}_n^T\bm{X}_n}{n})^{-1}\frac{\bm{X}_n^T}{n}\bm{\varepsilon}_n . 
								$
								As $ \frac{\gamma_n}{\sigma^2}
								{\bm\beta}_n^{0T}(\frac{\bm{X}_n^T\bm{X}_n}{n})^{-1}\frac{\bm{X_n}^T}{n}\bm{\varepsilon}_n \sim N(0, \bm V_n^\ast)$, where 
								\begin{eqnarray*} 
									\bm V_n^\ast & = &\left( \frac{\gamma_n}{\sigma^2}
									{\bm\beta}_n^{0T}(\frac{\bm{X}_n^T\bm{X}_n}{n})^{-1}\frac{\bm{X}_n^T}{n} \right) \sigma^2    \left(\frac{\gamma_n}{\sigma^2} 
									\frac{\bm{X_n}}{n} (\frac{\bm{X}_n^T\bm{X}_n}{n})^{-1} {\bm\beta}_n^0 \right) 
									=   \frac{\gamma_n^2}{n \sigma^2}  {\bm\beta}_n^{0T} \bm\Gamma_n\bm D_n^{-1}\bm\Gamma_n^{T}  {\bm\beta}_n^0   \\
									&\leq &   \frac{\gamma_n^2}{n \sigma^2} {\bm\beta}_n^{0T} \bm\Gamma_n \text{diag}\{1/d_{\min}, \cdots, 1/d_{\min} \} \bm\Gamma_n^{T}  {\bm\beta}_n^0  
									\leq  \frac{\gamma_n^2}{n \sigma^2} \frac{1}{d_{\min}} ||{\bm\beta}_n^{0}||^2  
									\rightarrow 0 ,
								\end{eqnarray*}
								so 
								\begin{eqnarray}\label{equantion-te2}
									\lim\limits_{p_n, n\rightarrow\infty} -\frac{1}{p_n} {\bm m}_n^T\bm V_n^{-1}(\frac{\bm{X}_n^T\bm{X}_n}{n}+\frac{\gamma_n\bm{I}_n}{n})^{-1}\frac{\bm{X}_n^T}{n}\bm{\varepsilon}_n  = 0.
								\end{eqnarray}

								According to  (\ref{chi2/p}), (\ref{formulat}), (\ref{equation-te1}) and (\ref{equantion-te2}),  we have 
								\begin{eqnarray} \label{limit-formula}
									\lim\limits_{p_n, n\rightarrow\infty} \frac{1}{p_n} (\hat{\bm\beta}_n-{\bm\beta}_n^0)^T\bm V_n^{-1}(\hat{\bm\beta}_n-{\bm\beta}_n^0) = 1. 
								\end{eqnarray}
								So 
								{  \begin{eqnarray}
										&& 1 = \lim\limits_{p_n, n\rightarrow\infty} \frac{1}{p_n} (\hat{\bm\beta}_n-{\bm\beta}_n^0)^T\bm V_n^{-1}(\hat{\bm\beta}_n-{\bm\beta}_n^0) \nonumber \\
										&=& \lim\limits_{p_n, n\rightarrow\infty} \frac{1}{p_n} (\hat{\bm\beta}_n-{\bm\beta}_n^0)^T 
										\frac{n}{\sigma^2} \bm\Gamma_n  \text{diag}\{(d_1+\gamma_n/n)^{2}/d_1,\cdots,(d_{p_n}+\gamma_n/n)^{2}/d_{p_n} \} \nonumber \\ &&
										\bm\Gamma_n^{T} 
										(\hat{\bm\beta}_n-{\bm\beta}_n^0) \nonumber \\
										&\geq & \lim\limits_{p_n, n\rightarrow\infty} \frac{1}{p_n} (\hat{\bm\beta}_n-{\bm\beta}_n^0)^T 
										\frac{n}{\sigma^2} \bm\Gamma_n  \text{diag}\{(d_1+\gamma_n/n),\cdots,(d_{p_n}+\gamma_n/n) \} \bm\Gamma_n^{T} 
										(\hat{\bm\beta}_n-{\bm\beta}_n^0) \nonumber \\
										&= & \lim\limits_{p_n, n\rightarrow\infty} \frac{1}{p_n} (\hat{\bm\beta}_n-{\bm\beta}_n^0)^T 
										\bm\Sigma^{-1}_n(\hat{\bm\beta}_n-{\bm\beta}_n^0). \nonumber
									\end{eqnarray}}
									Hence then
									$(\hat{{\bm\beta}}_n-{\bm\beta}_n^0)^T\bm\Sigma_n^{-1}(\hat{{\bm\beta}}_n-{\bm\beta}_n^0) = p_n \frac{1}{p_n} (\hat{\bm\beta}_n-{\bm\beta}_n^0)^T \bm\Sigma^{-1}_n(\hat{\bm\beta}_n-{\bm\beta}_n^0) \leq c_n$,  if   $c_n/p_n\rightarrow\infty$ and $\gamma_n=o(n)$, together with assumptions   \ref{a_error}-\ref{a_min}, the true parameter is contained in the region with probability tending to 1. 
								\end{proof}

								\begin{lemma}\label{lemma2-yan}
									Under  assumptions  \ref{a_error}-\ref{a_min},   and $\gamma_n=o(n)$,   
									the posterior mean $\hat{{\bm\beta}}_n =(\bm{X}_n^T\bm{X}_n+\gamma_n\bm{I}_n)^{-1}(\bm{X}_n^T\bm{Y}_n) $   has the property: $\frac{n}{p_n}||\hat{{\bm\beta}}_n - {\bm\beta}_n^0||^2 = O(1)$, or $\sqrt{\frac{n}{p_n}}(\hat{{ \beta}}_{nj} - {\ \beta}^0_{nj})=O(1)$ for $j=1,\cdots,p$. 
								\end{lemma}
								\begin{proof}
									Following  (\ref{limit-formula}) in the proof of Lemma \ref{lemma1}, we have $\lim\limits_{p_n, n\rightarrow\infty} \frac{1}{p_n} (\hat{\bm\beta}_n-{\bm\beta}_n^0)^T\bm V_n^{-1}(\hat{\bm\beta}_n-{\bm\beta}_n^0) = \lim\limits_{p_n, n\rightarrow\infty} \frac{1}{p_n} \frac{n}{\sigma^2}
									(\hat{\bm\beta}_n-{\bm\beta}_n^0)^T\bm\Gamma_n  \text{diag}\{(d_1+\gamma_n/n)^{2}/d_1,\cdots,(d_{p_n}+\gamma_n/n)^{2}/d_{p_n} \} \bm\Gamma_n^{T}  (\hat{\bm\beta}_n-{\bm\beta}_n^0) = 1$.  So  as $p_n, n\rightarrow\infty$,  we have 
									\begin{eqnarray*}
										&& \frac{n d_{\min}}{p_n\sigma^2} (\hat{\bm\beta}_n-{\bm\beta}_n^0)^T (\hat{\bm\beta}_n-{\bm\beta}_n^0) 
										= 
										\frac{n}{p_n\sigma^2} (\hat{\bm\beta}_n-{\bm\beta}_n^0)^T  \bm\Gamma_n\text{diag} \{d_{\min},\cdots,d_{\min}\}\bm\Gamma_n^T (\hat{\bm\beta}_n-{\bm\beta}_n^0)  \\
										&\leq &  \frac{n}{p_n\sigma^2} (\hat{\bm\beta}_n-{\bm\beta}_n^0)^T  \bm\Gamma_n\text{diag} \{d_{1},\cdots,d_{p_n}\}\bm\Gamma_n^T (\hat{\bm\beta}_n-{\bm\beta}_n^0)  \\
										&\leq &  \frac{n}{p_n \sigma^2}(\hat{\bm\beta}_n-{\bm\beta}_n^0)^T  \bm\Gamma_n \text{diag}\{(d_1+\gamma_n/n)^{2}/d_1,\cdots,(d_{p_n}+\gamma_n/n)^{2}/d_{p_n} \bm\Gamma_n^{T}  (\hat{\bm\beta}_n-{\bm\beta}_n^0) \rightarrow  1 . 
									\end{eqnarray*}
									Then $ \frac{n}{p_n} ||\hat{\bm\beta}_n-{\bm\beta}_n^0||^2 \leq \frac{\sigma^2}{d_{\min}}$, 
									i.e., $\frac{n}{p_n}||\hat{{\bm\beta}}_n - {\bm\beta}_n^0||^2 = O(1)$ or $\sqrt{\frac{n}{p_n}}(\hat{\beta}_{nj} - \beta^0_{nj})=O(1)$ for $j=1,\cdots,p$.  
									
									Note: This cannot ensure for every $j=1,\cdots,p$, $\sqrt{n}(\hat{\beta}_{nj}-\beta^0_{nj}) = O(1)$. 
									For example, if $\sqrt{n}(\hat{\beta}_{n1}-\beta^0_{n1} )=\sqrt{p_n}$, and all the other terms are zero, we'd still have $\frac{n}{p_n}||\hat{\bm \beta}_n - \bm\beta_n^0||^2 =  \frac{1}{p_n} \sum_{j=1}^{p_n} \left(\sqrt{n}(\hat{\beta}_{nj}-\beta^0_{nj})\right)^2= \frac{1}{p_n} (p_n+0) = 1 = O(1)$. 
									%
								\end{proof}

								\begin{lemma}\label{lemma3}
									Let   $\tilde{{\bm\beta}}_n$ be the solution to the optimization problem for the choice of $c_n$. 	Under  assumptions  \ref{a_error}-\ref{a_min},   if $\frac{c_n}{ p_n\log n}\rightarrow c$, where $0<c<\infty$, then $\frac{n}{p_n\log n}(\tilde{{\bm\beta}}_n-\hat{{\bm\beta}}_n)^T \frac{\bm\Sigma_n^{-1}}{n}(\tilde{{\bm\beta}}_n-\hat{{\bm\beta}}_n) = O(1)$, and $\left[\frac{n}{p_n \log n}(\tilde{{\bm\beta}}_n-\hat{{\bm\beta}}_n)^T\frac{\bm\Sigma_n^{-1}}{n}(\tilde{{\bm\beta}}_n-\hat{{\bm\beta}}_n)\right]^{-1} = O(1)$. 
								\end{lemma}
								\begin{proof}
									Suppose $\frac{c_n}{  p_n \log n }\rightarrow c$,  since the solution occurs on the boundary of the credible set, we have $n(\tilde{{\bm\beta}}_n-\hat{{\bm\beta}}_n)^T \frac{\bm\Sigma_n^{-1}}{n}(\tilde{{\bm\beta}}_n-\hat{{\bm\beta}}_n) =c_n$. Multiplying both sides by $\frac{1}{ p_n\log n }$, on  the right hand side we have $\frac{c_n}{ p_n\log n} \rightarrow c$.  
									So we have $\frac{n }{p_n \log n}(\tilde{{\bm\beta}}_n-\hat{{\bm\beta}}_n)^T \frac{\bm\Sigma_n^{-1}}{n}(\tilde{{\bm\beta}}_n-\hat{{\bm\beta}}_n) )\rightarrow c$. 
									Then $\frac{n }{p_n\log n }(\tilde{{\bm\beta}}_n-\hat{{\bm\beta}}_n)^T \frac{\bm\Sigma_n^{-1}}{n}(\tilde{{\bm\beta}}_n-\hat{{\bm\beta}}_n) = O(1)$, and $\left[\frac{n}{p_n \log n }(\tilde{{\bm\beta}}_n-\hat{{\bm\beta}}_n)^T\frac{\bm\Sigma_n^{-1}}{n}(\tilde{{\bm\beta}}_n-\hat{{\bm\beta}}_n)\right]^{-1} = O(1)$, since $0<c<\infty$. 
								\end{proof}

								\begin{lemma} \label{lemma4}
									Under  assumptions  \ref{a_error}-\ref{a_min},     if $\frac{c_n}{  p_n \log n }\rightarrow c$,  $\sqrt{\frac{n}{p_n}}(\tilde{\beta}_{nj} - \hat{\beta}_{nj})\rightarrow\infty$ can be true only for $j\in \mathcal{A}_n^0$. 
								\end{lemma}
								\begin{proof}
									From Lemma \ref{lemma3}, we have $\sqrt{\frac{n}{p_n}}(\tilde{\beta}_{nj} - \hat{\beta}_{nj}) \rightarrow\infty$ for some $j$. We now prove that  it cannot be true for $j\in \mathcal{A}_n^{0c}$. 
									
									Without loss of generality, we assume the true parameters are $\beta_{n1}^0 = \cdots = \beta_{nk}^0= 0$ and $\beta_{n k+1}^0= \cdots = \beta_{n p_n}^0 \neq 0$.   Assume $S(\bm\beta_n) = \sum_{j=1}^{p_n} |\hat\beta_{nj}|^{-2}|\beta_{nj}|$, then the  solution is the minimizer of $S({\bm\beta}_n)$ among those points within the given credible set. 
									
									Suppose $\tilde{{\bm\beta}}_n$ is the minimizer of $S({\bm\beta}_n)$, and suppose that  $1 \in \mathcal{A}_n^{0c}$, and $\sqrt{\frac{n}{p_n}}(\tilde{\beta}_{n1} - \hat{\beta}_{n1}) \rightarrow\infty$. Since $1 \in \mathcal{A}_n^{0c}$ or $\beta_{n1}^0= 0 $, it follows that $\sqrt{\frac{n}{p_n}}(\hat{\beta}_{n1}-0) = O(1)$. Hence, it must be that $\sqrt{\frac{n}{p_n}}\tilde{\beta}_{n1}\rightarrow\infty$.  Also  by Lemma \ref{lemma2-yan}$,  {\frac{n}{p_n}}|\hat{\beta}_{n1}|^2 = O(1)$. 
									So,  $\sqrt{\frac{p_n}{n}}S(\tilde{{\bm\beta}}_n) \geq \frac{\sqrt{\frac{n}{p_n}}|\tilde{{ \beta}}_{n1}|}{{\frac{n}{p_n}}|\hat{\beta}_{n1}|^2} \rightarrow \infty$.
									
									Now let $\tilde{{\bm\beta}}_n^\ast = \{0,\cdots,0,\tilde{{ \beta}}^\ast_{n k+1},\cdots,\tilde{{ \beta}}^\ast_{np_n} \}$ be the minimizer of $S({\bm\beta}_n)$, within the credible set,  with setting the first $k$ components to zero.  As shown in  the proof of Lemma 4  in  \cite{bondell2012consistent}, first,  for large $n$, such  form of $\tilde{\bm{\beta}}_n^\ast$ does exist within the credible region by Lemma \ref{lemma1}.   Next, we would show that 
									$\sqrt{\frac{p_n}{n}}S(\tilde{{\bm\beta}}_n^\ast)< \sqrt{\frac{p_n}{n}}S(\tilde{{\bm\beta}}_n)$ for large $n$, and hence achieve a contradiction.  
									
									By Lemma \ref{lemma2-yan} and Lemma \ref{lemma3},  for any $j \in \mathcal{A}_n^0$,   it follows that $\sqrt{\frac{n}{p_n}} (\hat{\beta}_{nj} - \beta_{nj}^0) = O(1)$ and $\sqrt{\frac{n}{p_n \log n}} (\hat{\beta}_{nj} - \tilde\beta_{nj}^\ast) = O(1)$.  
									Let $T_n=\sum_{j=k+1}^{p_n}| {\beta}_{nj}^0|^{-1}$. Then   we have 
									{\small 	\begin{eqnarray}
											&& |	S(\tilde{\bm\beta}_n^\ast)  -  	T_n|  = \left| \sum\limits_{j=k+1}^{p_n}    \frac{|\tilde{\beta}^\ast_{n j}|}{|\hat{\beta}_{n j}|^2} -    \sum\limits_{j=k+1}^{p_n}  | {\beta}_{nj}^0|^{-1}  \right|   \nonumber  \\
											&\leq &   \sum\limits_{j=k+1}^{p_n}  \left|  \frac{|\tilde{\beta}^\ast_{n j}|}{|\hat{\beta}_{n j}|^2} - | {\beta}_{nj}^0|^{-1}  \right| 
											= \sum\limits_{j=k+1}^{p_n}  \left|  \frac{|\tilde{\beta}^\ast_{n j}|}{|\hat{\beta}_{n j}|^2}   - | \hat{\beta}_{nj}|^{-1}   + | \hat{\beta}_{nj}|^{-1} - | {\beta}_{nj}^0|^{-1}  \right|  \nonumber   \\
											& =  &   \sum\limits_{j=k+1}^{p_n}   \left| \frac{|\tilde{\beta}^\ast_{n j}|  - | \hat{\beta}_{nj}| }{|\hat{\beta}_{n j}|^2}    + 
											\frac{   | {\beta}_{nj}^0| - | \hat{\beta}_{nj}|     }{ | \hat{\beta}_{nj} {\beta}_{nj}^0  |} 
											\right| 
											\leq   \sum\limits_{j=k+1}^{p_n}   \left(     \frac{ |\tilde{\beta}^\ast_{n j} -  \hat{\beta}_{nj}|}{|\hat{\beta}_{n j}|^2}    + 
											\frac{ |  \hat{\beta}_{nj}  - {\beta}_{nj}^0|   }{ | \hat{\beta}_{nj} {\beta}_{nj}^0  |} 
											\right)  \nonumber \\
											&=&   \sum\limits_{j=k+1}^{p_n}   \left(     \frac{  \sqrt{\frac{n}{p_n \log n}}|\tilde{\beta}^\ast_{n j} -  \hat{\beta}_{nj}|}{ \sqrt{\frac{n}{p_n \log n}}|\hat{\beta}_{n j}|^2}    + 
											\frac{  \sqrt{\frac{n}{p_n  }}|  \hat{\beta}_{nj}  - {\beta}_{nj}^0|   }{ \sqrt{\frac{n}{p_n }} | \hat{\beta}_{nj} {\beta}_{nj}^0  |} 
											\right) .  \nonumber 
										\end{eqnarray}
									}
									So 
									{\small \begin{eqnarray*}
											&& \sqrt{\frac{p_n}{n}} S(\tilde{\bm\beta}_n^\ast)  \leq  \sqrt{\frac{p_n}{n}}  T_n + 
											\sum\limits_{j=k+1}^{p_n}     \left(     \frac{  \sqrt{\frac{n}{p_n \log n}}|\tilde{\beta}^\ast_{n j} -  \hat{\beta}_{nj}|}{  {\frac{n}{p_n \sqrt{\log n }}}  |\hat{\beta}_{n j}|^2 }    + 
											\frac{  \sqrt{\frac{n}{p_n  }}|  \hat{\beta}_{nj}  - {\beta}_{nj}^0|   }{  {\frac{n}{p_n }}   | \hat{\beta}_{nj} {\beta}_{nj}^0  | } 
											\right)   \\
											&=&   \sqrt{\frac{p_n}{n}}  T_n +  \sum\limits_{j=k+1}^{p_n}    w_j^\ast  Z_j^\ast+   \sum\limits_{j=k+1}^{p_n}  w_jZ_j , 
										\end{eqnarray*}}
										where  $ Z_j^\ast =  \sqrt{\frac{n}{p_n \log n}}|\tilde{\beta}^\ast_{n j} -  \hat{\beta}_{nj}|$, $w_j^\ast =  ( {\frac{n}{p_n \sqrt{ \log n}}}  |\hat{\beta}_{n j}|^{2})^{-1} = 
										( {\frac{n}{p_n \sqrt{ \log n}}}   |{\beta}_{nj}^0|^{2} +o(1) )^{-1}$, $Z_j = \sqrt{\frac{n}{p_n  }}|  \hat{\beta}_{nj}  - {\beta}_{nj}^0|    $, and  $w_j =  ({\frac{n}{p_n }}   | \hat{\beta}_{nj} {\beta}_{nj}^0  |)^{-1} =  ({\frac{n}{p_n }}  |{\beta}_{nj}^0|^{2} + o(1))^{-1}$.  
										
										%
										
										Since for any $j \in \mathcal{A}_n^0$,   it follows that $Z_j^\ast= O(1)$ and $Z_j= O(1)$,  then 
										E$(Z_j^\ast) = O(1)$, E$(Z_j) = O(1)$,   Var$(Z_j^\ast)  = O(1)  $ and Var$(Z_j)  = O(1)  $. 
										Then  by assumption \ref{a_min}, 
										\[
										\text{E}( \sum\limits_{j=k+1}^{p_n}w_j^\ast Z_j^\ast  ) = O( \sum\limits_{j=k+1}^{p_n}   w_j^\ast )  =  O(1) , 
										\]
										and  since $ \sum\limits_{j=k+1}^{p_n}   w_j^{\ast2} \leq ( \sum\limits_{j=k+1}^{p_n}   w_j^\ast )^2  $,  we further have 
										\[
										\text{Var}( \sum\limits_{j=k+1}^{p_n} w_j^\ast Z_j^\ast ) = O( \sum\limits_{j=k+1}^{p_n}  w_j^{\ast 2} )  =   O(( \sum\limits_{j=k+1}^{p_n}   w_j^\ast )^2 ) = O(1) . 
										\]
										Similarly,  $\text{E}( \sum\limits_{j=k+1}^{p_n} w_j Z_j ) = O( \sum\limits_{j=k+1}^{p_n}   w_j )   =  O(1) $, and 
										$\text{Var}( \sum\limits_{j=k+1}^{p_n} w_j  Z_j   ) = O( \sum\limits_{j=k+1}^{p_n}   w_j^2  ) =  O(1)$. 
										Hence $\sum\limits_{j=k+1}^{p_n}  w_j  Z_j   $ and $\sum\limits_{j=k+1}^{p_n} w_j^\ast Z_j^\ast   $ are bounded  in probability. 
										
										Since  from Assumption \ref{a_min},  $\sqrt{\frac{p_n}{n}} T_n  = O(1)$,  then we have 
										$ \sqrt{\frac{p_n}{n}} 	S(\tilde{\bm\beta}_n^\ast)$ is bounded in probability.  But   $\sqrt{\frac{p_n}{n}}S(\tilde{{\bm\beta}}_n) \rightarrow\infty$ as shown above. 
										Hence  there exists large enough $n$ so that $ \sqrt{\frac{p_n}{n}} 	S(\tilde{\bm\beta}_n^\ast)< \sqrt{\frac{p_n}{n}} 	S(\tilde{\bm\beta}_n)$, and thus $\tilde{\bm\beta}_n$ with $\sqrt{\frac{n}{p_n}}(\tilde{\beta}_{n1} - \hat{\beta}_{n1})\rightarrow\infty$ cannot be the minimizer. 
										
										Therefore, $\sqrt{\frac{n}{p_n}}(\tilde{{ \beta}}_{nj} - \hat{{ \beta}}_{nj})\rightarrow\infty$ can be true only for $j\in \mathcal{A}_n^0$.
									\end{proof}

									\textbf{Proof of Theorem \ref{theorem_selection consistency}}
									\begin{proof}
										%
										%
										By assumption,   $p_n=o(n/\log n)$ and $c_n/(p_n\log n)\rightarrow c \in (0, \infty)$. 
										This implies that $c_n / p_n \rightarrow \infty$ and $c_n / n  \rightarrow 0$. 
										From Lemma \ref{lemma1}, it follows that if $c_n / p_n \rightarrow \infty$ then the true $\bm\beta^0_n$ is contained in the credible region with probability tending to one. 
										
										Furthermore,  if $c_n /  n  \rightarrow 0$, it must follow that the credible region itself is shrinking around $\hat{\bm\beta}_n$. 
										Since the solution occurs on the boundary of the credible set,    $(\tilde{{\bm\beta}}_n-\hat{{\bm\beta}}_n)^T \bm\Sigma_n^{-1}(\tilde{{\bm\beta}}_n-\hat{{\bm\beta}}_n)=(\tilde{{\bm\beta}}_n-\hat{{\bm\beta}}_n)^T \frac{(\bm{X}_n^T\bm{X}_n+\gamma_n \bm{I}_n)}{\hat\sigma_n^2}(\tilde{{\bm\beta}}_n-\hat{{\bm\beta}}_n) =c_n$, or equally, $(\tilde{{\bm\beta}}_n-\hat{{\bm\beta}}_n)^T (\bm{X}_n^T\bm{X}_n+\gamma_n \bm{I}_n)(\tilde{{\bm\beta}}_n-\hat{{\bm\beta}}_n)= c_n \hat\sigma_n^2$.  Multiplying both sides by $\frac{1}{n }$, on the right hand side, we have $\frac{c_n \hat\sigma_n^2}{n }\rightarrow  0 $. For the left side, we have $(\tilde{{\bm\beta}}_n-\hat{{\bm\beta}}_n)^T \frac{(\bm{X}_n^T\bm{X}_n+\gamma_n \bm{I}_n)}{n }(\tilde{{\bm\beta}}_n-\hat{{\bm\beta}}_n) $.	 So the left side also goes to zero. 
										Together with  Assumption \ref{a_eigenvalues} and $\gamma_n = o(n)$, then  the credible region is shrinking around 	$\hat{\bm\beta}_n$.  
										This implies that with probability tending to one, for all $j \in \mathcal{A}^0_n$, the credible region will be bounded away from zero in that direction. Hence we have that $P (\mathcal{A}^0_n  \cap \mathcal{A}_n^c) \rightarrow 0$.

										Next we will show that $P (\mathcal{A}^{0c}_n  \cap \mathcal{A}_n) \rightarrow0$, which will complete the proof.
										
										%
										Without loss of generality, assume $\beta_{n1}^0=0, \beta_{np_n}^0 \neq 0$ and by Lemma \ref{lemma4}, let $\sqrt{\frac{n}{p_n}}(\tilde{\beta}_{np_n} - \hat{\beta}_{np_n})\rightarrow\infty$. Denote $\sigma_{ij}$ as  the $ij^{th}$ element of $\bm\Sigma^{-1}_n$. 
										
										Assume  $\tilde{\beta}_{n1}\neq 0$, and we'll get a contradiction. The proof is following  the proof of Theorem 1 in \cite{bondell2012consistent}.  As $(\tilde{{\bm\beta}}_n-\hat{{\bm\beta}}_n)^T \bm\Sigma_n^{-1}(\tilde{{\bm\beta}}_n-\hat{{\bm\beta}}_n) =c_n$, then  $\tilde{\beta}_{np_n}$ can be   represented as a function of the remaining $p_n-1$ coefficients. Since $\tilde{\beta}_{n1}\neq 0$, the minimizer of $\sum_{j=1}^{p_n} \frac{1}{\hat{\beta}_{nj}^2 }|\beta_{nj}|$ with respect to $\beta_{n1}$  satisfies
										\begin{eqnarray}\label{eq11}
											\frac{1}{|\hat{\beta}_{n1}^2|} \text{sign}(\tilde{\beta}_{n1}) + \frac{1}{|\hat{\beta}_{np_n}^2|} \text{sign}(\tilde{\beta}_{np_n}) \left.\frac{\partial \beta_{np_n}}{\partial \beta_{n1}}\right\vert_{\tilde{\bm\beta}_n}=0.
										\end{eqnarray}
										Consider the first term on the left hand side. Since $\beta_{n1}^0=0$, we have $\hat{\beta}_{n1}\rightarrow0$, then $\frac{1}{|\hat{\beta}_{n1}^2|}\rightarrow \infty$.  Hence $\frac{1}{|\hat{\beta}_{n1}^2|} \text{sign}(\tilde{\beta}_{n1}) \rightarrow\infty$ if $\tilde{\beta}_{n1}\neq 0 $. For the second term, since $\beta_{np_n}^0\neq 0 $, $\frac{1}{|\hat\beta_{np_n}|^2} = O(1)$. 
										
										Differentiating $(\tilde{{\bm\beta}}_n-\hat{{\bm\beta}}_n)^T \bm\Sigma^{-1}_n(\tilde{{\bm\beta}}_n-\hat{{\bm\beta}}_n) =c_n$ with respect to $\beta_{n1}$ yields 
										\begin{eqnarray}\label{eq12}
											\left.\frac{\partial\beta_{np_n}}{\partial \beta_{n1}}\right\vert_{\tilde{{\bm\beta}}_n} =- \frac{\sum_{j=1}^{p_n}\sigma_{1j}(\tilde{\beta}_{nj} - \hat{\beta}_{nj}) }{\sum_{j=1}^{p_n}\sigma_{p_nj}(\tilde{\beta}_{nj} - \hat{\beta}_{nj})} = -
											\frac{\sum_{j=1}^{p_n}\sigma_{1j} \sqrt{\frac{n}{p_n\log n}}(\tilde{\beta}_{nj} - \hat{\beta}_{nj}) }{\sum_{j=1}^{p_n}\sigma_{p_nj} \sqrt{\frac{n}{p_n\log n}}(\tilde{\beta}_{nj} - \hat{\beta}_{nj})}, 
										\end{eqnarray}
										where   $\frac{c_n}{p_n\log n}\rightarrow c$. Note that by Lemma \ref{lemma3},  both numerator and denominator in (\ref{eq12}) are $O(1)$. Also by Lemma \ref{lemma4}, the denominator cannot be $0$, due to the presence of $(\tilde{\beta}_{np_n} - \hat{\beta}_{np_n})$.  Hence $\left.\frac{\partial\beta_{np_n}}{\partial \beta_{n1}}\right\vert_{\tilde{{\bm\beta}}_n}  = O(1)$. Then, $\frac{1}{|\hat{\beta}_{np_n}^2|} \text{sign}(\tilde{\beta}_{np_n}) \left.\frac{\partial \beta_{np_n}}{\partial \beta_{n1}}\right\vert_{\tilde{{\bm\beta}}_n}= O(1)$. Hence, the left   side of (\ref{eq11}) diverges, which yields a contradiction. 
									\end{proof}
									
									\textbf{Proof of Theorem \ref{theorem_gl_selection }}
									\begin{proof}
										For any  global-local shrinkage prior represented as (\ref{equation-globalLocal}), the prior precision  $\delta_j$ is   $  1/(w\xi_j)$. 
										The goal is to show that $\delta_j = o(n)$ for each  $j=1,\cdots, {p_n}$ with posterior  probability 1, 
										i.e.,  for  any $\epsilon>0$,    $P( \frac{\delta_{j} }{n}  \geq \epsilon | \bm{Y}_n) \rightarrow 0$  as $n\rightarrow\infty$. 
										
										When  the conditions of posterior consistency are satisfied, the  GL prior produces consistent posteriors, i.e, 
										for any $\epsilon>0$,  
										$P({\bm\beta}_n:||{\bm\beta}_n-{\bm\beta}_n^0||>\epsilon |\bm{Y}_n)\rightarrow 0$   as $ p_n, n\rightarrow \infty$.    Then the posterior mean, $\hat{{\bm\beta}}_n^{\text{GL}}$, satisfies  $P( ||\hat{{\bm\beta}}_n^{\text{GL}}-{\bm\beta}_n^0||>\epsilon  )\rightarrow 0$. 
										
										Also, since the ordinary least square estimator,  
										$\hat{{\bm\beta}}_n^{\text{OLS}}  = (\bm{X}_n^T\bm{X}_n)^{-1}  \bm{X}_n^T\bm{Y}_n $, is consistent,  it follows that $P( ||\hat{{\bm\beta}}_n^{\text{OLS}} - \hat{{\bm\beta}}_n^{\text{GL}}||>\epsilon  )\rightarrow 0$, or $ \hat{{\bm\beta}}_n^{\text{OLS}} -  \hat{{\bm\beta}}_n^{\text{GL}}  \rightarrow \bm{0}$. 
										
										Let $\bm\delta_n = (\delta_1,\cdots, \delta_{p_n})$, $\Delta_n = \text{diag}\{\delta_1, \cdots, \delta_{p_n}\}$ and $\pi(\bm\delta_n|\bm{Y}_n)$ be the posterior density function of $\bm\delta_n$, then  
										\[
										\hat{{\bm\beta}}_n^{\text{GL}} = \text{E}_{\bm\delta_n|\bm{Y}_n} [ \text{E} ( {\bm\beta}| \bm\delta_n, \bm{Y}_n) ] = \int_{\bm\delta_n} (\bm{X}_n^T\bm{X}_n +\Delta_n)^{-1} \bm{X}_n^T\bm{Y}_n \pi(\bm\delta_n | \bm{Y}_n) \, d\bm\delta_n .  
										\]
										Hence 
										\begin{eqnarray*}
											&&\hat{{\bm\beta}}_n^{\text{OLS}}   -  \hat{{\bm\beta}}_n^{\text{GL}}   
											= \int_{\bm\delta_n}  \left( (\frac{\bm{X}_n^T\bm{X}_n}{n})^{-1} - (\frac{\bm{X}_n^T\bm{X}_n}{n} + \frac{\Delta_n}{n})^{-1}
											\right)  \frac{\bm{X}_n^T\bm{Y}_n}{n}  \pi(\bm\delta_n | \bm{Y}_n) \, d\bm\delta_n \rightarrow \bm{0} . 
										\end{eqnarray*}
										Since as $n\rightarrow\infty$,  $   \bm{X}_n^T\bm{Y}_n /n \sim N( \frac{\bm{X}_n^T\bm{X}_n}{n} {\bm\beta_n^0}, \sigma^2\frac{\bm{X}_n^T\bm{X}_n}{n}) $ is a random variable independent of $\bm\delta_n$,  so  $P( \bm{X}_n^T\bm{Y}_n/n=0)=0$. Then  
										\begin{equation}\label{eq_temp1}
											\int_{\bm\delta_n}  \left( 
											(\frac{\bm{X}_n^T\bm{X}_n}{n})^{-1} - (\frac{\bm{X}_n^T\bm{X}_n}{n} + \frac{\Delta_n}{n})^{-1}\right)   \pi(\bm\delta_n | \bm{Y}_n) \, d\bm\delta_n \rightarrow \bm{0}.
										\end{equation}
										%
										Assume $\lambda_1,\cdots,\lambda_{p_n}$ are the eigenvalues of   
										$\frac{\bm{X}_n^T\bm{X}_n}{n} + \frac{\Delta_n}{n}$. 
										Also we have  $\frac{\bm{X}_n^T\bm{X}_n}{n} = \bm\Gamma_n\bm D_n\bm\Gamma_n^{T}$ where $\bm D_n = \text{diag}\{d_1,\cdots ,d_{p_n}\}$. 
										By Weyl's Inequalities, 
										$d_j + \frac{\delta_{\min}}{n}\leq \lambda_j\leq d_j+ \frac{\delta_{\max}}{n}$. 
										Together with \ref{a_eigenvalues},   then  $(\frac{\bm{X}_n^T\bm{X}_n}{n} + \frac{\Delta_n}{n})^{-1}$ is positive definite with probability 1. 
										
										(\ref{eq_temp1}) can also  be equivalently represented as 
										\[
										\int_{\bm\delta_n}  (\frac{\bm{X}_n^T\bm{X}_n}{n} + \frac{\Delta_n}{n})^{-1}  \left(  (\frac{\bm{X}_n^T\bm{X}_n}{n} + \frac{\Delta_n}{n})
										(\frac{\bm{X}_n^T\bm{X}_n}{n})^{-1} - {\bm I}_n \right)   \pi(\bm\delta_n | \bm{Y}_n) \, d\bm\delta_n \rightarrow \bm{0},
										\]
										or
										\[
										\int_{\bm\delta_n}  (\frac{\bm{X}_n^T\bm{X}_n}{n} + \frac{\Delta_n}{n})^{-1}  \left(   \frac{\Delta_n}{n} 
										\right)   (\frac{\bm{X}_n^T\bm{X}_n}{n})^{-1} \pi(\bm\delta_n | \bm{Y}_n) \, d\bm\delta_n \rightarrow \bm{0}.
										\]
										Since $ (\frac{\bm{X}_n^T\bm{X}_n}{n})^{-1}$ is positive definite for large enough $n$  by \ref{a_eigenvalues},  
										it follows that 
										\begin{equation}\label{eq:temp}
											\int_{\bm\delta_n}  (\frac{\bm{X}_n^T\bm{X}_n}{n} + \frac{\Delta_n}{n})^{-1}  \left(   \frac{\Delta_n}{n} 
											\right)    \pi(\bm\delta_n | \bm{Y}_n) \, d\bm\delta_n \rightarrow \bm{0}. 
										\end{equation}
										
										Also, for   $j=1,\cdots, p_n$, there exists such $M_n$ that 
										\begin{equation}\label{eq:00}
											\lim\limits_{n\rightarrow\infty} P(\frac{\delta_j}{n}>M_n) = 0. 
										\end{equation} 
										Now (\ref{eq:temp})	  	  and 	  (\ref{eq:00}) imply that 
										\[
										\int_{\bm\delta_n} \left(   \frac{\Delta_n}{n} \right)(\frac{\bm{X}_n^T\bm{X}_n}{n} + \frac{\Delta_n}{n})^{-1}  
										\left(   \frac{\Delta_n}{n} 
										\right)    \pi(\bm\delta_n | \bm{Y}_n) \, d\bm\delta_n \rightarrow \bm{0}. 
										\]
										Since $(\frac{\bm{X}_n^T\bm{X}_n}{n} + \frac{\Delta_n}{n})^{-1}  $ is positive definite, it must follow that 
										\[
										\lim\limits_{n\rightarrow\infty}P(\frac{\delta_1}{n}\geq\epsilon |{\bm Y}_n) =   0. 
										\]
									\end{proof}

									\textbf{Proof of  Theorem  \ref{theorem_normal tuning}} 
									\begin{proof}
										According to the assumption that  
										$\frac{\bm{X}^T\bm{X}}{n} =\bm\Gamma  \bm D\bm\Gamma^{T}$ where $\bm D = \text{diag}\{d_1,\cdots ,d_{p}\}$ with $d_1,\cdots, d_{p}$ denoting the eigenvalues. 
										Since $\bm\beta\sim N(0,\sigma^2/\gamma \bm I_{p})$, then  $\sqrt{\gamma/\sigma^2}\bm\Gamma^T\bm\beta\sim N(\bm 0, \bm I_{p})$.
										
										Also, 
										\begin{eqnarray*}
											R^2 = 1 - \frac{\sigma^2}{{\bm\beta}^T\frac{\bm{X}^T\bm{X}}{n}{\bm\beta} +\sigma^2} = 
											\frac{{\bm\beta}^T\frac{\bm{X}^T\bm{X}}{n}{\bm\beta}}{ {\bm\beta}^T\frac{\bm{X}^T\bm{X}}{n}{\bm\beta} +\sigma^2}
											=\frac{ \frac{\gamma}{\sigma^2}({\bm\beta}^T \bm\Gamma \bm D\bm\Gamma^T{\bm\beta})}{ \frac{\gamma}{\sigma^2}({\bm\beta}^T \bm\Gamma \bm D\bm\Gamma^T{\bm\beta}) +\gamma}
											= \frac{W}{W +  \gamma }, 
										\end{eqnarray*}
										where $W =\frac{\gamma}{\sigma^2}({\bm\beta}^T \bm\Gamma \bm D\bm\Gamma^T{\bm\beta}) = d_1Z_1^2 +\cdots + d_pZ_p^2$, where $Z_1,\cdots, Z_p$ are i.i.d. from $N(0,1)$, or $Z_1^2,\cdots, Z_p^2$ i.i.d. from 
										$\chi^2_1$. 
										Then 
										$W$  follows  a distribution with  density denoted as  $f_{W}(\cdot)$,  with mean 
										$\sum_{j=1}^p d_j$,  and variance $2\sum_{j=1}^p d_j^2$. 
										
										On the other hand, if  $R^2$ follows a Beta$(a,b)$ distribution, then  the density function of $W = \frac{\gamma R^2}{1-R^2}$   is  as follows:
										\[
										f_B(w) =   \frac{ \Gamma(a+b)}{ \Gamma(a)\Gamma(b)}\frac{w^{a-1}\gamma^b}{(\gamma+w)^{a+b}}, \ (w\geq 0 ) .
										\] 
										Thus, to get the solution of $\gamma$  to make the distribution of $R^2$ closest to the Beta distribution, or $f_W(\cdot)$ closest to $f_B(\cdot)$,  one needs to minimize the Kullback-Liebler directed divergence between them, which is given as below: 
										\begin{eqnarray*}
											&& \text{KL}(f_{W} | f_{  B})   =   \int_0^\infty f_{W}(x) \log\frac{f_{W}(x)}{f_B(x)} \, dx \\
											&=& \int_0^\infty f_{W}(x)  \log f_{W}(x) \, dx - \int_0^\infty f_{W}(x)  \log f_B(x) \, dx \\
											&=&  \int_0^\infty f_{W}(x)  \log f_{W}(x) \, dx  -  \int_0^\infty f_{W}(x) \log\left(  \frac{\Gamma(a+b)}{\Gamma(a)\Gamma(b)}\frac{x^{a-1}\gamma^b}{(\gamma+x)^{a+b}} \right) \, dx  \\
											&=& -b \log \gamma - (a-1)\text{E}[\log W  ]   + (a+b)\text{E}[\log(W+\gamma)]  +  C  \\
											&\approx& -b\log \gamma   + (a+b)\left( \log(\text{E}[W]+\gamma)   -  \frac{\text{Var}[W]}{2 (\text{E}[W]+\gamma)^2} \right) + C^\ast \\ 
											&=& - b\log \gamma  + (a+b)\left( \log(\sum_{j=1}^p d_j  +\gamma)   -  \frac{ 2\sum_{j=1}^p d_j^2}{2 (\sum_{j=1}^p d_j+\gamma)^2} \right) + C^\ast   , 
										\end{eqnarray*}
										where $C $ and $C^\ast$ are some  constant value with no relation to $\gamma$, and E($\cdot$) and Var($\cdot$) denote the expectation and variance  of the random variable $W$ with density $f_W(\cdot)$.  
										The derivation of the above formula replies on  the following facts: 
										(i) $\log x = \log x_0 + \frac{1}{x_0}(x-x_0) - \frac{(x-x_0)^2}{2x_0^2} + O( (x-x_0)^2)$;  
										(ii) When $p\rightarrow\infty$, a third derivative is small, so 
										a second order Taylor expansion around $x_0 = \text{E}[X]$  can be used to approximate $\text{E}[\log X]$: 
										\[ \text{E}[\log X] \approx \text{E}[\log x_0 ]  + \text{E}[\frac{1}{x_0}(X-x_0)]  -  \text{E}[\frac{(X-x_0)^2}{2x_0^2}]
										= \log(\text{E}[X])   -  \frac{\text{Var}[X]}{2 (\text{E}[X] )^2};  \]  
										(iii) Similarly, a second order Taylor expansion around $x_0' = x_0 + c = \text{E}[X] +c $  is used to approximate $\text{E}[\log(X+c)]$:
										\[ \text{E}[\log (X+c)] \approx \text{E}[\log (x_0+c)  ]  + \text{E}[\frac{X-x_0}{x_0+c}]  -  \text{E}[\frac{(X-x_0)^2}{2(x_0+c)^2}]
										= \log(\text{E}[X]+c)   -  \frac{\text{Var}[X]}{2 (\text{E}[X]+c )^2}. \]

										Then taking  the derivative of  $   \text{KL}(f_W | f_B) $, i.e., 
										$  \text{KL}'(f_W| f_B) = -\frac{b}{\gamma} + \frac{a+b}{\sum_{j=1}^p d_j+\gamma } + \frac{2(a+b)\sum_{j=1}^p d_j^2 }{(\sum_{j=1}^p d_j+\gamma)^3}$,  and letting it  equal to 0, we have 
										\begin{equation}\label{eq_taucube}
											\gamma^3 + \frac{2a-b}{a}(\sum_{j=1}^p d_j)\gamma^2 + \left( \frac{2(a+b)}{a} \sum_{j=1}^p d_j^2 + \frac{a-2b}{a}(\sum_{j=1}^p d_j)^2 \right) \gamma  - \frac{b}{a}( \sum_{j=1}^p d_j)^3 = 0 . 
										\end{equation}
										According to the conclusion in \cite{osler36easy},  if 
										$
										P = \frac{2a-b}{a}\sum_{j=1}^p d_j$,  
										$Q = \frac{2(a+b)}{a} \sum_{j=1}^p d_j^2 + \frac{a-2b}{a}(\sum_{j=1}^p d_j)^2$,   
										$R = -\frac{b}{a}( \sum_{j=1}^p d_j)^3$,  
										%
										$C = P^2/9 - Q/3, \ 
										A = PQ/6 - P^3/27 - R/2, \ 
										B = A^2 - C^3,  
										$
										and  $B\geq 0 $, then $\gamma = (A+\sqrt{B})^{1/3} + (A- \sqrt{B})^{1/3}  - P/3$ is the unique  real solution to (\ref{eq_taucube}). 
									\end{proof}

\end{document}